\documentclass[letterpaper,12pt]{article}

\usepackage{amssymb,amsmath,amsthm,amscd,latexsym,algorithm}
\usepackage{url}
\usepackage{fontenc}[T1]
\usepackage{lmodern}
\usepackage{luca}
\usepackage{times}
\usepackage{paralist}
\usepackage{fullpage}

\newtheorem{theorem}{Theorem}
\newtheorem{lemma}{Lemma}
\newtheorem{remark}{Remark}
\newtheorem{proposition}{Proposition}

\def\qed{\rule{0.4em}{1.4ex}}

\newcommand{\pat}{\omega}
\newcommand{\Paths}{\Omega}
\newcommand{\PA}{1}

\newcommand{\straa}{\sigma}
\newcommand{\Straa}{\Sigma}

\newcommand{\SA}{V_1}

\newcommand{\SR}{V_{P}}

\newcommand{\gamegraph}{G}

\newcommand{\winas}[1]{\langle \! \langle #1 \rangle\! \rangle_{\mathit{almost}} }

\newcommand{\waa}{\winas{1}}

\newcommand{\Prb}{\mathbb{P}}
\newcommand{\Inf}{\mathrm{Inf}}

\newcommand{\Parity}{{\mathrm{Parity}}}

\newcommand{\Buchi}{\textrm{B\"uchi}}

\newcommand{\attr}{\mathit{Attr}}

\newcommand{\Nats}{\mathbb{N}}
\newcommand{\nats}{\mathbb{N}}

\newcommand{\set}[1]{\{\: #1 \:\}}
\newcommand{\seq}[1]{\langle #1 \rangle}

\newcommand{\trans}{\delta}
\newcommand{\distr}{{\cal D}}

\newcommand{\Aa}{{\cal A}}

\newcommand{\slopefrac}[2]{\leavevmode\kern.1em
  \raise .5ex\hbox{\the\scriptfont0 #1}\kern-.1em
  /\kern-.15em\lower .25ex\hbox{\the\scriptfont0 #2}}

\begin{document}

\title{Average Case Analysis of the Classical Algorithm for \\ Markov Decision Processes with B\"uchi Objectives\footnotetext{A preliminary version appeared in the proceedings of 32nd IARCS Annual Conference on Foundations of Software Technology and Theoretical Computer Science (FSTTCS), 2012.\\\indent The research was supported by FWF Grant No P 23499-N23, FWF NFN Grant No S11407-N23 (RiSE), ERC Start grant (279307: Graph Games), and Microsoft faculty fellows award. Nisarg Shah is also supported by NSF grant CCF-1215883.}}

\author{Krishnendu Chatterjee\\IST Austria\\\small{\tt krish.chat@gmail.com} \and Manas Joglekar\\Stanford University\\\small{\tt manasrj@stanford.edu} \and Nisarg Shah\\Carnegie Mellon University\\\small{\tt nkshah@cs.cmu.edu}}
\date{}
\maketitle 

\begin{abstract}
We consider Markov decision processes (MDPs) with specifications given as B\"uchi (liveness) objectives, and examine the problem of computing the set of \emph{almost-sure} winning vertices such that the objective can be ensured with probability $1$ from these vertices. We study for the first time the average case complexity of the classical algorithm for computing the set of almost-sure winning vertices for MDPs with B\"uchi objectives. Our contributions are as follows:  First, we show that for MDPs with constant out-degree the expected number of iterations is at most logarithmic and the average case running time is linear (as compared to the worst case linear number of iterations and quadratic time complexity). Second, for the average case analysis over all MDPs we show that the expected number of iterations is constant and the average case running time is linear (again as compared to the worst case linear number of iterations and quadratic time complexity). Finally we also show that when all MDPs are equally likely, the probability that the classical algorithm requires more than a constant number of iterations is exponentially small.
\end{abstract}

\section{Introduction}

%tcs
In this work, we consider the qualitative analysis of Markov decision processes with B\"uchi (liveness) objectives, and establish optimal bounds for the average case complexity. We start by briefly describing the model and the objectives, then the significance of qualitative analysis, followed by the previous results, and finally our contributions.

\noindent{\bf Markov decision processes.} 
\emph{Markov decision processes (MDPs)} are standard models for 
probabilistic systems that exhibit both probabilistic 
and nondeterministic behavior~\cite{Howard}, and widely used 
in verification of probabilistic systems~\cite{BaierBook,PRISM}.
MDPs have been used to model and solve control problems for stochastic systems~\cite{FV97}: 
there, nondeterminism represents the freedom of the controller to choose a 
control action, while the probabilistic component of the behavior describes the 
system response to control actions. 
MDPs have also been adopted as models for concurrent probabilistic 
systems~\cite{CY95},  probabilistic systems operating in open 
environments~\cite{SegalaT}, under-specified probabilistic 
systems~\cite{BdA95}, and applied in diverse domains~\cite{PRISM}.
A \emph{specification} describes the set of desired behaviors of
the system, which in the verification and control of stochastic systems is 
typically an $\omega$-regular set of paths. % in the MDP. 
The class of $\omega$-regular languages extends classical regular languages to 
infinite strings, and provides a robust specification language to express
all commonly used specifications, such as safety, liveness, fairness, etc~\cite{Thomas97}. 
Parity objectives are a canonical way to define such $\omega$-regular specifications.
Thus MDPs with parity objectives provide the theoretical framework to 
study problems such as the verification and control of stochastic systems.

\smallskip\noindent{\bf Qualitative and quantitative analysis.} 
The analysis of MDPs with parity objectives can be classified into  
qualitative and quantitative analysis. 
Given an MDP with parity objective, the \emph{qualitative analysis} 
asks for the computation of the set of vertices from
where the parity objective can be ensured with probability~1 
(almost-sure winning).
%%in which the controller can achieve the parity objective with probability~1 
%%(almost-surely).
The more general \emph{quantitative analysis} asks for the computation 
of the maximal (or minimal) probability at each state with which the 
controller can satisfy the parity objective. 

\smallskip\noindent{\bf Importance of qualitative analysis.} 
The qualitative analysis of MDPs is an important problem in verification 
that is of interest independent of the quantitative analysis problem.
There are many applications where we need to know whether the correct 
behavior arises with probability~1.
For instance, when analyzing a randomized embedded scheduler, we are
interested in whether every thread progresses with probability~1~\cite{CdAFMR13}.
Even in settings where it suffices to satisfy certain specifications with 
probability $p<1$, the correct choice of $p$ is a challenging problem, due 
to the simplifications introduced during modeling.
For example, in the analysis of randomized distributed algorithms it is 
quite common to require correctness with probability~1 
(see, e.g., \cite{PSL00,KNP_PRISM00,Sto02b}). 
Furthermore, in contrast to quantitative analysis, 
qualitative analysis is robust to numerical perturbations and modeling errors in the 
transition probabilities, and consequently the algorithms for qualitative analysis are 
combinatorial.
%%as compared to the numerical algorithms for the quantitative analysis. 
Finally, for MDPs with parity objectives, the best known algorithms and all
algorithms used in practice first perform the qualitative analysis, and then 
perform a quantitative analysis on the result of the qualitative 
analysis~\cite{CY95,luca-thesis,ChaThesis,Cha12,CdAH05,CJH04}. 
Thus qualitative analysis for MDPs with parity objectives is one of the most 
fundamental and core problems in verification of probabilistic systems.

\smallskip\noindent{\bf Previous results.} The qualitative analysis for MDPs with 
parity objectives is achieved by iteratively applying solutions of the 
qualitative analysis of MDPs with B\"uchi objectives~\cite{CY95,luca-thesis,CJH04}.
The qualitative analysis of an MDP with a parity objective with $d$ priorities 
can be achieved  by $O(d)$ calls to an algorithm for qualitative analysis of MDPs with B\"uchi
objectives, and hence we focus on %the qualitative analysis of 
MDPs with B\"uchi objectives.
The qualitative analysis problem for MDPs with B\"uchi objectives has been widely studied. 
The classical algorithm for %qualitative analysis for MDPs with B\"uchi objectives
the problem was given in~\cite{CY95,luca-thesis}, and the worst case running time of the classical 
algorithm is $O(n \cdot m)$ time, where $n$ is the number of vertices, 
and $m$ is the number of edges of the MDP.
Many improved algorithms have also been given in the literature, such as~\cite{CJH03,CH11,CH12,CH14,CHJS13},
and several special cases have also been studied~\cite{CL13}, 
and the current best known worst case complexity of the problem is $O(\min\set{n^2,m\cdot \sqrt{m}})$.
Moreover, there exists a family of MDPs where the running time of the improved 
algorithms match the above bound. 
While the worst case complexity of the problem has been studied, to the best of our knowledge
the average case complexity of none of the algorithms has been studied in the literature.

\smallskip\noindent{\bf Our contribution.} In this work we study for the first 
time the average case complexity of the qualitative analysis of MDPs with 
B\"uchi objectives.
Specifically we study the average case complexity of the classical algorithm 
for the following two reasons: 
First, the classical algorithm is very simple and appealing as it iteratively uses 
solutions of the standard graph reachability and alternating graph reachability algorithms, 
and can be implemented efficiently by symbolic algorithms.
Second, while more involved algorithms that improve the worst case complexity have been 
proposed~\cite{CJH03,CH11,CH12,CH14,CHJS13}, it has also been established in~\cite{CH12,CHJS13} 
that there are simple variants of the involved algorithms that require at most a linear running time 
in addition to the time of the classical algorithm, %that never require more than a linear time in addition to the time of the classical algorithm, 
and hence the average case complexity of these 
variants is no more than the average case complexity of the classical algorithm.
We study the average case complexity of the classical algorithm and establish 
that 
%tcs - as 
compared to the quadratic worst case complexity, 
the average case complexity is linear.
Our main contributions are summarized below:

\begin{compactenum}

\item \emph{MDPs with constant out-degree.} 
We first consider MDPs with constant out-degree.
In practice, MDPs often have constant out-degree: for example, 
see~\cite{dAR07} for MDPs with large state space but constant number of 
actions, or~\cite{FV97,Puterman94} for examples from inventory management 
where MDPs have constant number of actions (the number of actions correspond to 
the out-degree of MDPs). 
% We consider MDPs with constant out-degree with constants $d_{\min}$ and 
% $d_{\max}$ such that every vertex $v$ has out-degree $d_v$ with 
% $d_{\min} \leq d_v \leq d_{\max}$ (i.e., out-degree bounded by 
% $d_{\min}$ and $d_{\max}$), and every set of neighbours of 
% size $d_v$ are equally likely and independent.
We consider MDPs where the out-degree of every vertex is fixed and given. The out-degree 
of a vertex $v$ is $d_v$ and there are constants $d_{\min}$ and $d_{\max}$ such that 
for every $v$ we have $d_{\min} \le d_v \le d_{\max}$. Moreover, every subset of the set of 
vertices of size $d_v$ is equally likely to be the neighbour set of $v$, independent of
the neighbour sets of other vertices.
We show that the expected number of iterations of the classical algorithm is 
at most logarithmic ($O(\log n)$), and the average case running time is 
linear ($O(n)$) (as compared to the worst case linear number of iterations and 
quadratic $O(n^2)$ time complexity of the classical algorithm, and the current 
best known $O(n\cdot \sqrt{n})$ worst case complexity). The average case complexity 
of this model implies the same average case complexity for several related models 
of MDPs with constant out-degree. For further discussion on this, see Remark~\ref{remark_related_model}.

\item \emph{MDPs in the Erd\"os-R\'enyi model.} 
To consider the average case complexity over all MDPs, we consider MDPs 
where the underlying graph is a random directed graph according to the classical 
Erd\"os-R\'enyi random graph model~\cite{ER}. 
We consider random graphs $\calg_{n,p}$, over $n$ vertices where each edge 
exists with probability $p$ (independently of other edges). 
To analyze the average case complexity over all MDPs with all graphs equally 
likely,
%tcs - added comma
 we need to consider the $\calg_{n,p}$ model with $p=\frac{1}{2}$ 
(i.e., each edge is present or absent with equal probability, and thus 
all graphs are considered equally likely).
We show a stronger result (than only $p=\frac{1}{2}$) that if 
$p \geq \frac{c\cdot \log(n)}{n}$, 
for some constant $c>2$, then the expected number of iterations of the 
classical algorithm is constant ($O(1)$), and the average case running time is 
linear (again as compared to the worst case linear number of iterations and 
quadratic time complexity).
Note that we obtain that the average case (when $p=\frac{1}{2}$) running time 
for the classical algorithm is linear over all MDPs (with all graphs equally 
likely) as a special case of our results for 
$p \geq \frac{c\cdot \log(n)}{n}$, for any constant $c>2$, since 
$\frac{1}{2} \geq \frac{3\cdot \log(n)}{n}$ for $n \ge 17$.
Moreover we show that when $p=\frac{1}{2}$ (i.e., all graphs are equally 
likely), the probability that the classical algorithm will require more than 
constantly many iterations is exponentially small in $n$ (less than 
$\left(\frac{3}{4}\right)^n$).
\end{compactenum}

\smallskip\noindent{\em Implications of our results.} 
We now discuss several implications of our results.
First, since we show that the classical algorithm has average case linear time
complexity, it follows that the average case complexity of qualitative analysis
of MDPs with B\"uchi objectives is linear time. 
Second, since qualitative analysis of MDPs with B\"uchi objectives is a more
general problem than reachability in graphs (graphs are a special case 
of MDPs and reachability objectives are a special case of B\"uchi objectives), 
the best average case complexity that can be achieved is linear. 
Hence our results for the average case complexity are tight.
Finally, since for the improved algorithms there are simple variants that 
never require more than linear time as compared to the classical algorithm it 
follows that the improved algorithms also have average case linear time 
complexity.
Thus we complete the average case analysis of the algorithms for the 
qualitative analysis of MDPs with B\"uchi objectives.
In summary our results show that the classical algorithm (the most simple and 
appealing algorithm) has excellent and optimal (linear-time) average case 
complexity as compared to the quadratic worst case complexity.

\smallskip\noindent{\em Technical contributions.} The two key technical 
difficulties to establish our results are as follows:
(1)~Though there are many results for random undirected graphs, for the average 
case analysis of the classical algorithm we need to analyze random directed graphs;
and
(2)~in contrast to other results related to random undirected graphs that prove 
results for almost all vertices, the classical algorithm stops 
%tcs
only 
when all vertices 
satisfy a certain reachability property; and hence we need to prove results for 
all vertices (as compared to almost all vertices).
In this work we set up novel recurrence relations to estimate the expected 
number of iterations, and the average case running time of the classical algorithm.
Our key technical results prove many interesting inequalities related to the 
recurrence relation for  reachability properties of random directed graphs 
to establish the desired result.
We believe the new interesting results related to reachability properties 
we establish for random directed graphs will find future applications in 
average case analysis of other algorithms related to verification.
%%Detailed proofs omitted due to space restriction are presented in the 
%%full version attached as appendix.

\section{Definitions}
\label{section:definition}

\noindent{\bf Markov decision processes (MDPs).}
A \emph{Markov decision process (MDP)} $\gamegraph =((V, E), (\SA,\SR),\trans)$ 
consists of a directed graph $(V,E)$, a partition $(\SA$,$\SR)$ of the 
\emph{finite} set $V$ of vertices, and a probabilistic transition function 
$\trans$: $\SR \rightarrow \distr(V)$, where $\distr(V)$ denotes the 
set of probability distributions over the vertex set~$V$. 
The vertices in $\SA$ are the {\em player-$\PA$\/} vertices, where player~$\PA$
decides the successor vertex, and the vertices in $\SR$ are the 
{\em probabilistic (or random)\/} vertices, 
where the successor vertex is chosen according to the probabilistic transition
function~$\trans$. 
We assume that for $u \in \SR$ and $v \in V$, we have $(u,v) \in E$ 
iff $\trans(u)(v) > 0$, and we often write $\trans(u,v)$ for $\trans(u)(v)$. 
For a vertex $v\in V$, we write $E(v)$ to denote the set 
$\set{u \in V \mid (v,u) \in E}$ of possible out-neighbours, 
and $|E(v)|$ is the out-degree of $v$.
For technical convenience we assume that every vertex in the graph 
$(V,E)$ has at least one outgoing edge, i.e., $E(v)\neq \emptyset$ for all 
$v \in V$.

\smallskip\noindent{\bf Plays, strategies and probability measure.}
An infinite path, or a \emph{play}, of the graph $\gamegraph$ is an 
infinite 
sequence $\pat=\seq{v_0, v_1, v_2, \ldots}$ of vertices such that 
$(v_k,v_{k+1}) \in E$ for all $k \in \Nats$. 
We write $\Paths$ for the set of all plays, and for a vertex $v \in V$, 
we write $\Paths_v\subseteq\Paths$ 
for the set of plays that start from the vertex~$v$.
A \emph{strategy} for  player~$\PA$ is a function 
$\straa$: $V^*\cdot \SA \to \distr(V)$ that chooses the probability 
distribution over the successor vertices for all finite sequences 
$\vec{w} \in V^*\cdot \SA$ of vertices 
ending in a player-1 vertex (the sequence represents a prefix of a play).
A strategy must respect the edge relation: for all $\vec{w} \in V^*$ and 
$u \in \SA$, if $\straa(\vec{w}\cdot u)(v) >0$, then $v \in E(u)$.
\begin{comment}
A strategy is \emph{deterministic (pure)} if it chooses a unique successor
for all histories (rather than a probability distribution), otherwise
it is \emph{randomized}.
Player~$\PA$ follows the strategy~$\straa$ if in each player-1 
move, given that the current history of the game is
$\vec{w} \in V^* \cdot \SA$, she chooses the 
next vertex according to  $\straa(\vec{w})$.
We denote by $\Straa$ the set of all strategies for player~$\PA$.
A \emph{memoryless} player-1 strategy does not depend on the history of 
the play but only on the current vertex; i.e., for all $\vec{w},\vec{w'} \in 
V^*$ and for all $v \in \SA$ we have 
$\straa(\vec{w} \cdot v) =\straa(\vec{w}'\cdot v)$.
A memoryless strategy can be represented as a function 
$\straa$: $\SA \to \distr(V)$, and a pure memoryless 
strategy can be represented as $\straa: \SA \to V$.
\end{comment}
Let $\Straa$ denote the set of all strategies. Once a starting vertex  $v \in V$ and a strategy $\straa \in \Straa$ is fixed, 
the outcome of the MDP is a random walk $\pat_v^{\straa}$ for which the
probabilities of events are uniquely defined, where an \emph{event}  
$\Aa \subseteq \Paths$ is a measurable set of plays. 
For a vertex $v \in V$ and an event $\Aa\subseteq\Paths$, we write
$\Prb_v^{\straa}(\Aa)$ for the probability that a play belongs 
to $\Aa$ if the game starts from the vertex $v$ and player~1 follows
the strategy $\straa$.
%%For a measurable function $f:\Paths \to \reals$ we denote by 
%%$\Exp_s^{\straa,\strab}[f]$ the \emph{expectation} of the function
%%$f$ under the probability measure $\Prb_s^{\straa,\strab}(\cdot)$.

\smallskip\noindent{\bf Objectives.}
We specify \emph{objectives} for the player~1 by providing
a set of \emph{winning} plays $\Phi \subseteq \Omega$.
We say that a play $\pat$ {\em satisfies} the objective
$\Phi$ if $\pat \in \Phi$.
We consider \emph{$\omega$-regular objectives}~\cite{Thomas97},
specified as parity conditions.
We also consider the special case of B\"uchi objectives.
\begin{compactitem}

\item
 \emph{B\"uchi objectives.} Let $B \subseteq V$ be a set of B\"uchi vertices.
  For a play $\pat = \seq{v_0, v_1, \ldots} \in \Omega$,
  we define $\Inf(\pat) =
  \set{v \in V \mid \mbox{$v_k = v$ for infinitely many $k$}}$
  to be the set of vertices that occur infinitely often in~$\pat$.
  The B\"uchi objectives require that some vertex of $B$ be visited 
  infinitely often, and defines the set of winning plays 
  $\Buchi(B)=\set{\pat \in \Paths \mid \Inf(\pat) \cap B \neq \emptyset}$.

\item
  \emph{Parity objectives.}
  For $c,d \in \nats$, we write $[c..d] = \set{c, c+1, \ldots, d}$.
  Let $p$: $V \to [0..d]$ be a function that assigns a \emph{priority}
  $p(v)$ to every vertex  $v \in V$, where $d \in \Nats$.
  The \emph{parity objective} is defined as
  $\Parity(p)=
  \set{\pat \in \Paths \mid
  \min\big(p(\Inf(\pat))\big) \text{ is even }}$.
  In other words, the parity objective requires that the minimum 
  priority visited infinitely often is even.
  In the sequel we will use $\Phi$ to denote parity objectives.
\end{compactitem}

\smallskip\noindent{\em Qualitative analysis: almost-sure winning.}
Given a player-1 objective~$\Phi$, a strategy $\straa\in\Sigma$ is  
\emph{almost-sure winning} for player~1 from the vertex $v$ 
if $\Prb_v^{\straa} (\Phi) =1$.
The \emph{almost-sure winning set} $\waa(\Phi)$ for player~1 is the set 
of vertices from which player~1 has an almost-sure winning strategy.
The qualitative analysis of MDPs corresponds to the computation of 
the almost-sure winning set for a given objective $\Phi$.
%It follows from the results of~\cite{CY95,luca-thesis} that for all MDPs and 
%all reachability and parity objectives, if there is an almost-sure 
%winning strategy, then there is a memoryless almost-sure winning strategy.
%The qualitative analysis of MDPs with parity objectives is achieved by 
%iteratively applying the solutions of qualitative analysis for MDPs with 
%B\"uchi objectives~\cite{luca-thesis,CJH04}, and hence in this 
%work we will focus on qualitative analysis for B\"uchi objectives.

%\vspace{-0.5em}
%\begin{theorem}[\cite{CY95,luca-thesis}]\label{thrm_memless}
%For all MDPs $G$, and all reachability and parity objectives $\Phi$, 
%there exists a pure memoryless strategy $\straa_*$ such that for all 
%$v \in \waa(\Phi)$ we have $\Prb_v^{\straa_*}(\Phi)=1$.
%\end{theorem}
%\vspace{-0.5em}

\begin{remark}[Implication for parity objectives]
The almost-sure winning set for MDPs with parity objectives can be 
computed using $O(d)$ calls to compute the almost-sure winning 
set of MDPs with B\"uchi objectives~\cite{CJH04,CY95,luca-thesis,ChaThesis,Cha12,CdAH05}.
Hence we focus on the qualitative analysis of MDPs with 
B\"uchi objectives. 
We will establish that the average case complexity is linear for 
B\"uchi objectives which implies an $O(m \cdot d)$ upper bound 
on the average case complexity for the qualitative analysis of MDPs
with parity objectives, where $m$ is the number of edges.
\end{remark}

\smallskip\noindent{\bf Algorithm for qualitative analysis.}
The algorithms for qualitative analysis for MDPs do not depend
on the transition function, but only on the graph $G=((V,E),(\SA,\SR))$. 
We now describe the classical algorithm for the qualitative analysis 
of MDPs with B\"uchi objectives. 
%tcs - and 
The algorithm requires the notion of 
random attractors.

\smallskip\noindent{\bf Random attractor.} 
Given an MDP $G$, let $U \subseteq V$ be a subset of vertices. 
The \emph{random attractor} $\attr_{P}(U)$ is defined 
%tcs - inductively 
as follows: 
$X_0=U$, and for $i\geq 0$, let  
$X_{i+1}=X_i \cup \set{v \in \SR \mid E(v) \cap X_i \neq \emptyset } \cup 
\set{v \in \SA \mid E(v) \subseteq X_i}$.
In other words, $X_{i+1}$ consists of (a)~vertices in $X_i$, 
(b)~probabilistic vertices that have at least one edge to $X_i$, and
(c)~player-1 vertices, 
%whose all successors are 
whose every successor is 
in $X_i$.
Then $\attr_{P}(U)=\bigcup_{i\geq 0} X_i$. 
Observe that the random attractor is equivalent to the alternating reachability 
problem (reachability in AND-OR graphs).

\smallskip\noindent{\bf Classical algorithm.} 
The classical algorithm for MDPs with B\"uchi objectives is a simple
iterative algorithm, and every iteration uses graph reachability and
alternating graph reachability (random attractors).
Let us denote the MDP in iteration $i$ by $G^i$ with vertex set $V^i$.
Then in iteration $i$ the algorithm executes the following 
steps: 
(i)~computes the set $Z^i$  of vertices that can reach the set of 
B\"uchi vertices $B \cap V^i$ in $G^i$; 
(ii)~let $U^i=V^i\setminus Z^i$ be the set of remaining vertices; 
if $U^i$ is empty, then the algorithm stops and outputs $Z^i$ as the set of 
almost-sure winning vertices, and otherwise removes 
$\attr_{P}(U^i)$ from the graph, and continues to iteration $i+1$.
The classical algorithm requires $O(n)$ iterations, where $n=|V|$,
and each iteration requires $O(m)$ time, where $m=|E|$.
Moreover the above analysis is tight, i.e., there exists a family of MDPs where
the classical algorithm requires $\Omega(n)$ iterations, and total time 
$\Omega(n\cdot m)$. 
Hence $\Theta(n\cdot m)$ is the tight worst case complexity of the classical
algorithm for MDPs with B\"uchi objectives.
In this work we consider the average case analysis of the classical 
algorithm.

\section{Average Case Analysis for MDPs with Constant Out-degree}
In this section we consider the average case analysis of the number of 
iterations and the running time of the classical algorithm for computing 
the almost-sure winning set for MDPs with B\"uchi objectives on the families of 
graphs with constant out-degree (out-degree of every vertex fixed and bounded by two 
constants $d_{\min}$ and $d_{\max}$).

\smallskip\noindent{\em Family of graphs and results.} 
We consider families of graphs where the vertex set $V$ ($|V|=n$), 
the target set of B\"uchi vertices $B$ ($|B|=t$), 
and the out-degree $d_v$ of each vertex $v$ is fixed across the whole family. 
The only varying component is the edges of the graph; for each vertex $v$, 
every set of vertices of size $d_v$ is equally likely to be the neighbour set of $v$, 
independent of neighbours of other vertices. 
Finally, there exist constants $d_{\min}$ and $d_{\max}$ such that 
$d_{\min} \le d_v \le d_{\max}$ for all vertices $v$. 
We will show the following for this family of graphs:  
(a)~if the target set $B$ has size more than $30 \cdot x\cdot \log(n)$,
where $x$ is the number of distinct degrees, (i.e., $t\geq 30\cdot x\cdot \log(n)$), then the expected number of iterations 
is $O(1)$ and the average running time is $O(n)$; and 
(b)~if the target vertex set $B$ has size at most $30\cdot x \cdot \log(n)$, 
then the expected number of iterations required is $O(\log(n))$ 
and average running time is $O(n)$.

\smallskip\noindent{\em Notation.}
We use $n$ and $t$ for the total number of vertices and the size of the target 
set, respectively.
We will denote by $x$ the number of distinct out-degrees. Let $d_i$, for $1 \le i \le x$, be the distinct out-degrees.
Since for all vertices $v$ we have $d_{\min} \le d_v \le d_{\max}$, 
it follows that we have $x \le d_{\max}-d_{\min}+1$.
Let $a_i$ be the number of vertices with degree $d_i$ and $t_i$ be the number 
of target (B\"uchi) vertices with degree $d_i$.

\smallskip\noindent{\em The event $R(k_1,k_2,...,k_x)$.}
The \emph{reverse reachable set} of the target set $B$ is the set of vertices 
$u$ such that there is a path in the graph from $u$ to a vertex $v \in B$.
Let $S$ be any set comprising of $k_i$ vertices of degree $d_i$, for $1 \leq i \leq x$. 
We define $R(k_1,k_2,...,k_x)$ as the probability of the event that all vertices of $S$ 
can reach $B$ via a path that lies entirely in $S$. Due to symmetry between vertices, 
this probability only depends on $k_i$, for $1 \le i \le x$ and is independent of $S$ itself.\footnote{This holds because the outdegrees of vertices in $S$ are fixed, but their neighbors are chosen randomly.} 
For ease of notation, we will sometimes denote the event itself by $R(k_1,k_2,...,k_x)$. 
We will investigate the reverse reachable set of $B$, which contains $B$ itself. 
Recall that $t_i$ vertices in $B$ have degree $d_i$, and hence we are interested in the case when $k_i \geq t_i$
for all $1 \leq i \leq x$.

%\section{Summation Equation}
%Consider a random graph $G$ having $n$ vertices, $a_i$ vertices of degree $d_i$, $1 \le i \le x$ and all $d_i$'s being distinct. Let the target set be comprised of $t_i$ vertices of degree $d_i$. 
%Let $S$ be the reverse reachable set of the target set, i.e. the set of vertices which can reach a target vertex. Let $|S|=k$. Let $S$ be comprised of $k_i$ vertices of degree $d_i$, $1 \le i \le x$. Note that $k_i \ge t_i, \forall i$. We are interested in $R(k_1,k_2,...,k_x)$.\\

Consider a set $S$ of vertices that is the reverse reachable set, and 
let $S$ be composed of $k_i$ vertices of degree $d_i$ and of size $k$, i.e., 
$k=|S|=\sum_{i=1}^x k_i$. 
Since $S$ is the reverse reachable set, it follows that for all vertices 
$v$ in $V\setminus S$, there is no edge from $v$ to a vertex in $S$ 
(otherwise there would be a path from $v$ to a target vertex and then 
$v$ would belong to $S$). 
Thus there are no incoming edges from $V\setminus S$ to $S$. 
Thus for each vertex $v$ of $V \setminus S$, all its neighbours must lie 
in $V\setminus S$ itself. 
This happens with probability
$\prod_{i\in[1,x], a_i \neq k_i} {\left(\frac{\binom{n-k}{d_i}}{\binom{n}{d_i}}\right)}^{a_i - k_i}$,
since in $V \setminus S$ there are $a_i-k_i$ vertices with degree $d_i$ and the 
size of $V\setminus S$ is $n-k$ (recall that $[1,x]=\{1,2,\dots,x\}$). Note that when $a_i \neq k_i$, there is at least one vertex of 
degree $d_i$ in $V \setminus S$ that has all its neighbours in $V \setminus S$ and hence $n-k \ge d_i$.
For simplicity of notation, we skip mentioning $a_i \neq k_i$ and substitute the term by $1$ where $a_i = k_i$.
The probability that each vertex in $S$ can reach a target vertex is 
$R(k_1,k_2,...,k_x)$. 
Hence the probability of $S$ being the reverse reachable set is given by:
\begin{equation*}
\prod_{i=1}^{x} {\left(\frac{\binom{n-k}{d_i}}{\binom{n}{d_i}}\right)}^{a_i-k_i} \cdot R(k_1,k_2,...,k_x)
\end{equation*}
There are $\prod_{i=1}^{x} \binom{a_i-t_i}{k_i-t_i}$ possible ways of choosing 
$k_i \ge t_i$ vertices (since the target set is contained) out of $a_i$. 
Notice that the terms are $1$ where $a_i = k_i$.
The value $k$ can range from $t$ to $n$ and exactly one of these subsets of $V$ 
will be the reverse reachable set. 
So the sum of probabilities of this happening is $1$. Hence we have:
\begin{equation}
\label{eqn:recurrence}
1 = \sum_{k=t}^{n} \ \ \sum_{\sum k_i = k, t_i \le k_i \le a_i} \left(\prod_{i=1}^{x} \binom{a_i-t_i}{k_i-t_i} \cdot {\left(\frac{\binom{n-k}{d_i}}{\binom{n}{d_i}}\right)}^{a_i-k_i}\right) \cdot R(k_1,k_2,...,k_x)
\end{equation} 
Let 
\[
\begin{array}{rcl}
a_{k_1,k_2,...,k_x} & = &
\displaystyle 
\left(\prod_{i=1}^{x} \binom{a_i-t_i}{k_i-t_i} \cdot {\left(\frac{\binom{n-k}{d_i}}{\binom{n}{d_i}}\right)}^{a_i-k_i}\right)\cdot R(k_1,k_2,...,k_x); \\[2.5ex]
%\qquad \text{and} 
\alpha_{k} & = & 
\displaystyle 
\sum_{\sum k_i = k, t_i \le k_i \le a_i} a_{k_1,k_2,...,k_x}.
\end{array}
\]

%nisarg - Put the main lemma of this whole part first. Then main lemma for small k, then its lemmas and so on.
Thus, $a_{k_1, k_2,...,k_x}$ is the probability that the reverse reachable set has exactly $k_i$ vertices of degree $d_i$ for $1 \leq i \leq x$, and $\alpha_{k}$ is the probability that the reverse reachable set has exactly $k$ vertices.
 
Our goal is to show that for $30 \cdot x \cdot \log(n) \le k \le n-1$, the 
value of $\alpha_k$ is very small; i.e., we want to get an upper bound on 
$\alpha_k$.
Note that two important terms in $\alpha_k$ are ${\left(\binom{n-k}{d_i}/\binom{n}{d_i}\right)}^{a_i-k_i}$ and $R(k_1,k_2,\dots,k_x)$. 
Below we get an upper bound for both of them. 
Firstly note that when $k$ is small, for any set $S$ comprising of $k_i$ vertices of degree $d_i$ for $1 \le i \le x$ and $|S|=k$, 
the event $R(k_1,k_2,\dots,k_x)$ requires each non-target vertex of $S$ to have an edge inside $S$.
Since $k$ is small and all vertices have constant out-degree spread randomly over the entire graph, 
this is highly improbable. We formalize this intuitive argument in the following lemma.

\begin{lemma}[Upper bound on $R(k_1,k_2,\ldots,k_x)$]\label{lemm_upper_bound}
For $k \le n-d_{\max}$
\[
R(k_1,k_2,\ldots,k_x) \le \prod_{i=1}^{x}  
{\left(1- \left(1-\frac{k}{n-d_i}\right)^{d_i}\right)}^{k_i - t_i}
\leq \prod_{i=1}^{x} {\left(\frac{d_i \cdot k}{n-d_{\max}}\right)}^{k_i - t_i}.
\]
\end{lemma}
\begin{proof}
Let $S$ be the given set comprising of $k_i$ vertices of degree $d_i$, for $1 \le i \le x$. 
Then for every non-target vertex of $S$, for it to be reachable to a target vertex via a path in $S$, 
it must have at least one edge inside $S$.
This gives the following upper bound on $R(k_1,k_2,...,k_x)$.
\[
R(k_1,k_2,...,k_x) \le \prod_{i=1}^{x} {\left(1-\frac{\binom{n-k}{d_i}}{\binom{n}{d_i}}\right)}^{k_i - t_i} 
\]
We have the following inequality for all $d_i$, $1 \leq i \leq x$: 
\[
\frac{\binom{n-k}{d_i}}{\binom{n}{d_i}} 
=  
\prod_{j=0}^{d_i-1} \left(1-\frac{k}{n-j}\right) \\[1ex]
\ge 
{\left(1-\frac{k}{n-d_i}\right)}^{d_i}
\ge  
1-\frac{d_i \cdot k}{n-d_i} 
\]
The first inequality follows by replacing $j$ with $d_i \geq j$, and the second
inequality follows from standard binomial expansion. 
Using the above inequality in the bound for $R(k_1,k_2,\ldots,k_x)$ 
we obtain
\[
R(k_1,k_2,...,k_x) \leq \prod_{i=1}^{x} {\left(1- \left(1-\frac{k}{n-d_i}\right)^{d_i}\right)}^{k_i - t_i} 
\leq \prod_{i=1}^{x} {\left(\frac{d_i \cdot k}{n-d_i}\right)}^{k_i - t_i} \le \prod_{i=1}^{x} {\left(\frac{d_i \cdot k}{n-d_{\max}}\right)}^{k_i - t_i}
\]
The result follows.
\hfill\qed
\end{proof}

%Note that we have a loose and a strict bound on $R(k_1,k_2,\dots,k_x)$. We use the loose bound whenever it is sufficient 
%and switch to the strict bound whenever required.
Now for ${\left(\binom{n-k}{d_i}/\binom{n}{d_i}\right)}^{a_i-k_i}$, we give an upper bound. First notice that when $a_i \neq k_i$, there is at least one vertex of degree $d_i$ outside the reverse reachable set and it has all its edges outside the reverse reachable set. Hence, the size of the reverse reachable set (i.e. $n-k$) is at least $d_i$. Thus, $\binom{n-k}{d_i}$ is well defined.

\begin{lemma}\label{lemm_bound_2} For any $1 \le i \le x$ such that $a_i \neq k_i$, we have 
${\left(\frac{\binom{n-k}{d_i}}{\binom{n}{d_i}}\right)}^{a_i-k_i} 
\le {\left(1-\frac{k}{n}\right)}^{d_i\cdot(a_i-k_i)}$.
\end{lemma}
\begin{proof}
We have
\[
{\left(\frac{\binom{n-k}{d_i}}{\binom{n}{d_i}}\right)}^{a_i-k_i}  
= {\left(\prod_{j=0}^{d_i-1} \left(1-\frac{k}{n-j}\right)\right)}^{a_i-k_i} 
 \le {\left(1-\frac{k}{n}\right)}^{d_i\cdot(a_i-k_i)}
\]
The inequality follows since $j\geq 0$ and we replace $j$ by~0 in the 
denominator.
The result follows.
\hfill\qed
\end{proof}

Next we simplify the expression of $\alpha_k$ by taking care of the summation.
\begin{lemma}\label{lemm_bound_3}
The probability that the reverse reachable set is of size exactly $k$ is 
$\alpha_k$, and 
\[
\alpha_k \leq 
n^x \cdot \max_{\sum k_i=k, t_i \leq k_i \leq a_i} a_{k_1,k_2,\ldots, k_x}.
\]
\end{lemma}
\begin{proof}
The probability that the reverse reachable set is of size exactly $k$ 
is given by 
\[
\alpha_k=\sum_{\sum k_i = k, t_i \le k_i \le a_i} \left(\prod_{i=1}^{x} \binom{a_i-t_i}{k_i-t_i} \cdot {\left(\frac{\binom{n-k}{d_i}}{\binom{n}{d_i}}\right)}^{a_i-k_i}\right) \cdot R(k_1,k_2,...,k_x),
\]
(refer to Equation~\ref{eqn:recurrence}).
Since  
\[
\alpha_k=\sum_{\sum k_i = k, t_i \le k_i \le a_i} a_{k_1,k_2,\ldots,k_x},
\]
and there are $x$ distinct degree's and $n$ vertices, the number of different 
terms in the summation is at most $n^x$.
Hence 
\[
\alpha_k \leq n^x \cdot \max_{\sum k_i=k, t_i \leq k_i \leq a_i} a_{k_1,k_2,\ldots, k_x}.
\]
The desired result follows.
\hfill\qed
\end{proof}

Now we proceed to achieve an upper bound on $a_{k_1,k_2,\ldots,k_x}$. 
First of all, intuitively if $k$ is small, then $R(k_1,k_2,\dots,k_x)$ is very small (this can be derived easily from Lemma~\ref{lemm_upper_bound}). 
On the other hand, consider the case when $k$ is very large.
In this case there are very few vertices that cannot reach the target set. Hence they must have all their edges within them, which again has very low probability.
Note that different factors that bind $\alpha_k$ depend on whether $k$ is small or large. This suggests we should consider these cases separately.
Our proof will consist of the following case analysis of the size $k$ of
the reverse reachable set: 
%tcs - (1)~when $30\cdot x\cdot \log(n) \le k\le c_1 \cdot n$ is \emph{small} (for some constant $c_1>0$); 
%tcs - (2)~when $c_1 \cdot n \le k \le c_2 \cdot n$ is \emph{large}
%tcs - (3)~when $c_2\cdot n \le k \le n-d_{\min}-1$ is \emph{very large}. 
(1)~Small $k$: $30\cdot x\cdot \log(n) \le k\le c_1 \cdot n$ for some constant $c_1>0$, 
(2)~Large $k$: $c_1 \cdot n \le k \le c_2 \cdot n$ for all constants $c_2 \geq c_1>0$, and 
(3)~Very large $k$: $c_2\cdot n \le k \le n-d_{\min}-1$  for some constant $c_2>0$. 
The analysis of the constants will follow from the proofs. 
Note that since the target set $B$ (with $|B|=t$) is a subset of its reverse reachable set, 
%tcs - we have 
the case 
$k < t$ is infeasible. Hence in all the three cases, we will only consider $k \ge t$.
We first consider the case when $k$ is small.

\subsection{Small $k$: $30\cdot x\cdot \log(n) \le k \le c_1 n$} 
In this section we will consider the case when 
$30\cdot x\cdot \log(n) \le k \le c_1 \cdot n$ for some constant $c_1>0$. 
Note that this case only occurs when $t \le c_1 \cdot n$ (since $k \ge t$). We will assume this throughout this section.
We will prove that there exists a constant $c_1>0$ such that for all $30 \cdot x \cdot \log(n) \le k \le c_1 \cdot n$ 
the probability ($\alpha_k$) that the size of the reverse 
reachable set is $k$ is bounded by $\frac{1}{n^2}$.
Note that we already have a bound on $\alpha_k$ in terms of $a_{k_1,k_2,\ldots,k_x}$ (Lemma~\ref{lemm_bound_3}). 
We use continuous upper bounds of the discrete functions in $a_{k_1,k_2,\ldots,k_x}$ to convert it into a form that is easy to analyze.
Let 
$$ b_{k_1,k_2,...,k_x} = \prod_{i=1}^{x} \left(\frac{e \cdot (a_i-t_i)}{k_i-t_i}\right)^{k_i-t_i}\cdot 
e^{-\frac{k}{n} \cdot d_i \cdot (a_i-k_i)} \cdot {\left(\frac{d_i\cdot k}{n-d_{\max}}\right)}^{k_i-t_i},$$
where $e$ is Euler's number (the base of the natural logarithm). 
\begin{lemma}\label{lemm_bound_small_1}
We have $a_{k_1,k_2,\ldots,k_x} \leq b_{k_1,k_2,\ldots,k_x}$.
\end{lemma}
\begin{proof} We have
\begin{eqnarray*}
a_{k_1,k_2,...,k_x} && = \left(\prod_{i=1}^{x} \binom{a_i-t_i}{k_i-t_i} \cdot {\left(\frac{\binom{n-k}{d_i}}{\binom{n}{d_i}}\right)}^{a_i-k_i}\right) \cdot R(k_1,k_2,...,k_x) \\
&& \le \prod_{i=1}^{x} \binom{a_i-t_i}{k_i-t_i} \cdot {\left(1-\frac{k}{n}\right)}^{d_i\cdot (a_i-k_i)} \cdot {\left(\frac{d_i \cdot k}{n-d_{\max}}\right)}^{k_i-t_i} \\
&& \le \prod_{i=1}^{x} \left(\frac{e \cdot (a_i-t_i)}{k_i-t_i}\right)^{k_i-t_i} \cdot e^{-\frac{k}{n} d_i (a_i-k_i)} \cdot  {\left(\frac{d_i\cdot k}{n-d_{\max}}\right)}^{k_i-t_i}
\end{eqnarray*}
The first inequality follows from Lemma~\ref{lemm_upper_bound} and Lemma~\ref{lemm_bound_2}.
The second inequality follows from the first inequality of Proposition~\ref{prop_approx} (in technical appendix) and
the fact that $1-x \leq e^{-x}$.
\hfill\qed 
\end{proof}

\smallskip\noindent{\em Maximum of $b_{k_1,k_2,\ldots,k_x}$.}
Next we show that $b_{k_1,k_2,...,k_x}$ drops exponentially as a function of $k$. 
Note that this is the reason for the logarithmic lower bound on $k$ in this section.
To achieve this we consider the 
maximum possible value achievable by $b_{k_1,k_2,\ldots,k_x}$.
Let $\partial_{k_i} b_{k_1,k_2,...,k_x}$ denote the change in $b_{k_1,k_2,\ldots,k_x}$ 
due to change in $k_i$.
For fixed $\sum_{i=1}^{x} k_i = k$, it is known that $b_{k_1,k_2,...,k_x}$ is maximized when for all $i$ and $j$ we have
$\partial_{k_i} b_{k_1,k_2,...,k_x} = \partial_{k_j} b_{k_1,k_2,...,k_x}$.
We have
\begin{equation*}
\partial_{k_i} b_{k_1,k_2,...,k_x} = b_{k_1,k_2,...,k_x} \cdot 
\left(\frac{d_i \cdot k}{n} + \log\left(\frac{d_i\cdot k}{n-d_{\max}}\right) + 
\log\left(\frac{a_i-t_i}{k_i-t_i} \right)\right)
\end{equation*}
Thus, for maximizing $b_{k_1,k_2,...,k_x}$, for all $i$ and $j$ we must have 
\begin{eqnarray*}
&& \frac{d_i \cdot k}{n} + \log\left(\frac{d_i \cdot k}{n-d_{\max}}\right) + \log\left(\frac{a_i-t_i}{k_i-t_i}\right) 
= \frac{d_j \cdot k}{n} + \log\left(\frac{d_j \cdot k}{n-d_{\max}}\right) + \log\left(\frac{a_j-t_j}{k_j-t_j}\right) \\
&& \Rightarrow \frac{k_i-t_i}{(a_i-t_i) \cdot \frac{d_i \cdot k}{n-d_{\max}} \cdot e^{d_i \cdot k/n}} = \frac{k_j-t_j}{(a_j-t_j) \cdot \frac{d_j \cdot k}{n-d_{\max}} \cdot e^{d_j\cdot k/n}} \\
&& \Rightarrow \frac{k_i-t_i}{(a_i-t_i) \cdot d_i \cdot e^{d_i\cdot k/n}} = \frac{k_j-t_j}{(a_j-t_j) \cdot d_j \cdot e^{d_j \cdot k/n}} 
\end{eqnarray*}
This implies that for all $i$ we have
\begin{eqnarray*}
&& \frac{k_i-t_i}{(a_i-t_i)\cdot  d_i \cdot e^{d_i \cdot k/n}} = \frac{k-t}{\sum_{i=1}^{x} (a_i-t_i) \cdot d_i \cdot e^{d_i \cdot k/n}} \\
&& \Rightarrow 
k_i-t_i = \frac{(a_i-t_i) \cdot d_i \cdot e^{d_i \cdot k/n}}{\sum_{i=1}^{x} (a_i-t_i) \cdot d_i \cdot e^{d_i \cdot k/n}} \cdot (k-t)
\end{eqnarray*}

\begin{lemma}\label{lemm_small_2}
Let $L = \sum_{i=1}^{x} (a_i-t_i)\cdot d_i \cdot e^{d_i\cdot k/n}$.
We have 
\[
b_{k_1,k_2,...,k_x}  
\le \left(\frac{L}{n-d_{\max}}\right)^{-t} \cdot 
\left( \frac{L}{n-d_{\max}} \cdot e^{1-\frac{\sum_{i=1}^{x} d_i \cdot (a_i-t_i)}{n}} \right)^k
\]
\end{lemma}
\begin{proof} 
The argument above shows that the maximum of 
$b_{k_1,k_2,...,k_x}$ is achieved when for all $1\le i \le x$ we have
$k_i-t_i = \frac{(a_i-t_i) \cdot d_i \cdot e^{d_i \cdot k/n}}{L} \cdot (k-t)$.
Now, plugging the values in $b_{k_1,k_2,...,k_x}$, we get 
\begin{eqnarray*}
b_{k_1,k_2,...,k_x} && = 
\prod_{i=1}^{x} \left(\frac{e \cdot (a_i-t_i)}{k_i-t_i}\right)^{k_i-t_i} 
\cdot e^{-\frac{k}{n} \cdot d_i \cdot (a_i-k_i)} \cdot 
{\left(\frac{d_i \cdot k}{n-d_{\max}}\right)}^{k_i-t_i} \\
%%%%%%%%%%%%%%%%%%%%%%%
&& \le \prod_{i=1}^{x} \left(\frac{e \cdot L}{d_i \cdot e^{d_i \cdot k/n} \cdot (k-t)}\right)^{k_i-t_i} 
\cdot e^{-\frac{k}{n} \cdot d_i \cdot (a_i-k_i)} \cdot 
{\left(\frac{d_i \cdot k}{n-d_{\max}}\right)}^{k_i-t_i} \\  %%\qquad \text{{\bf Q1.}}\\
%%%%%%%%%%%%%%%%%%%%%%%
&& = \prod_{i=1}^{x} \left(\frac{L}{n-d_{\max}}\right)^{k_i-t_i} 
\cdot \left( e^{(k_i-t_i)}\cdot e^{-d_i \cdot (k/n) \cdot (k_i-t_i)} \cdot e^{-\frac{k}{n} \cdot d_i \cdot (a_i-k_i)} \right)\cdot 
{\left(\frac{d_i \cdot k}{d_{i} \cdot (k-t)}\right)}^{k_i-t_i} \\ 
&& \qquad \text{(Rearranging denominators of first and third term, gathering powers of $e$ together)}\\[2ex]
%%%%%%%%%%%%%%%%%%%%%%%
&& = \left(\frac{L}{n-d_{\max}}\right)^{\sum_{i=1}^x(k_i-t_i)} 
\cdot \left( e^{\sum_{i=1}^x(k_i-t_i)}\cdot e^{-\sum_{i=1}^x d_i \cdot (k/n) \cdot (a_i-t_i)} \right)\cdot 
{\left(\frac{k}{(k-t)}\right)}^{\sum_{i=1}^x (k_i-t_i)} \\ 
&& \qquad \text{(Product is transformed to sum in exponent)}\\[2ex]
%%%%%%%%%%%%%%%%%%%%%%%
&& = \left(\frac{L}{n-d_{\max}}\right)^{(k-t)} 
\cdot \left( e^{(k-t)}\cdot e^{-(k/n)\cdot \sum_{i=1}^x d_i \cdot (a_i-t_i)} \right)\cdot 
{\left(1+ \frac{t}{(k-t)}\right)}^{(k-t)} \\ 
&& \qquad \text{(As $\sum_{i=1}^x k_i-t_i=k-t$)}\\[2ex]
%%%%%%%%%%%%%%%%%%%%%%%
&& \le \left(\frac{L}{n-d_{\max}}\right)^{k-t} \cdot e^{k-t} \cdot 
e^{-k/n \cdot \sum_{i=1}^{x} d_i \cdot (a_i-t_i)} \cdot e^{t} \\[2ex]
& & \qquad \text{(Since $1+x \leq e^x$ we have $\left(1+\frac{t}{k-t}\right) \leq e^{\frac{t}{k-t}}$)} \\
%%%%%%%%%%%%%%%%%%%%%%%
&& = 
\left(\frac{L}{n-d_{\max}}\right)^{-t} \cdot 
\left( \frac{L}{n-d_{\max}} \cdot e^{1-\frac{\sum_{i=1}^{x} d_i \cdot (a_i-t_i)}{n}} \right)^k  \\
%%%%%%%%%%%%%%%%%%%%%%%
& & \qquad \text{(Arranging in powers by $t$ and $k$).}\\
%&& = \left(e \frac{\sum_{i=1}^{x} d_i (a_i-t_i) e^{d_i k/n}}{n-d_{max}}\right)^{-t} \left( \frac{\sum_{i=1}^{x} d_i (a_i-t_i) e^{d_i k/n}}{n-d_{max}} e^{1-\frac{\sum_{i=1}^{x} d_i (a_i-t_i)}{n}} \right)^k  \\
\end{eqnarray*}
%{\bf Q1: should there not be a term $k-t$ in the denominator of first term.
%Q2. How does the term $k$ vanishes for the third term. 
%If you want a inequality it would probably be fine to cancel $k$ with $k-t$.}
The desired result follows.
\hfill\qed
\end{proof}

%tcs
We now establish an upper bound on each term in the bound of Lemma~\ref{lemm_small_2}. First, we consider the term $\frac{L}{n-d_{\max}} \cdot e^{1-\frac{\sum_{i=1}^{x} d_i \cdot (a_i-t_i)}{n}}$. 
\begin{lemma}\label{lemm_small_3}
Let $n$ be sufficiently large and let $c_1 \le \frac{0.04}{d_{\max}}$.
Then for all $k \leq c_1 \cdot n$ we have 
$\left( \frac{L}{n-d_{\max}} \cdot e^{1-\frac{\sum_{i=1}^{x} d_i \cdot (a_i-t_i)}{n}} \right) \leq \frac{9}{10}$.
\end{lemma}
\begin{proof}
We have the following inequality: %%see that
%%Now, for any general $t$, we need to make the base of the second term less than some constant, say $z$. %%{\bf Question why do we need to do that.}
\begin{eqnarray*}
\left( \frac{L}{n-d_{\max}} \cdot e^{1-\frac{\sum_{i=1}^{x} d_i \cdot (a_i-t_i)}{n}} \right) = 
&& \frac{\sum_{i=1}^{x} d_i\cdot (a_i-t_i)\cdot e^{d_i \cdot k/n}}{n-d_{\max}}
\cdot e^{1-\frac{\sum_{i=1}^{x} d_i \cdot (a_i-t_i)}{n}} \\[2ex]
\le && 
\frac{ e^{d_{\max} \cdot c_1} }{n-d_{\max}} \cdot \left(\sum_{i=1}^{x} d_i\cdot (a_i-t_i)\right) \cdot e^{1-\frac{\sum_{i=1}^{x} d_i \cdot (a_i-t_i)}{n}} 
\\[2ex]
 && \qquad \text{($d_i \leq d_{\max}$ and $k\leq c_1\cdot n$)}\\[1.5ex]
\le && 
e^{d_{\max} \cdot c_1} \cdot \frac{n}{n-d_{\max}} \cdot \frac{\sum_{i=1}^{x} d_i\cdot (a_i-t_i)}{n} \cdot e^{1-\frac{\sum_{i=1}^{x} d_i \cdot (a_i-t_i)}{n}}
\\[2ex]
 && \qquad \text{(multiplying numerator and denominator with $n$)}\\[1.5ex]
= && e^{d_{\max} \cdot c_1} \cdot \frac{n}{n-d_{\max}} \cdot \frac{d}{e^{d-1}}
% \le && \frac{ e^{d_{\max} \cdot c_1} \cdot n}{n-d_{\max}} \cdot e^{-\frac{1}{e}} \qquad \left(f(d) = \frac{d}{e^{d-1}} \text{ has its maximum at $d = \frac{1}{e}$}\right) \\
% \le && e^{d_{\max} \cdot c_1 - \frac{1}{e}} \cdot \frac{1}{0.9} \qquad \text{ (for large enough n)} 
\end{eqnarray*}
Here, 
$$
d = \frac{1}{n} \cdot \sum_{i=1}^x d_i \cdot (a_i-t_i) \ge d_{\min} \cdot \frac{n-t}{n} \ge d_{\min} \cdot (1-c_1) \ge 1
$$
The last inequality follows because $c_1 \le 0.5$ and $d_{\min} \ge 2$. Since $f(d) = d/e^{d-1}$ is a decreasing function for $d \ge 1$, we have $f(d) \le f(d_{\min} \cdot (1-c_1))$. Thus,
\begin{align*}
\frac{ e^{d_{\max} \cdot c_1} \cdot n}{n-d_{\max}} \cdot \frac{d}{e^{d-1}} & \le e^{d_{\max} \cdot c_1} \cdot \frac{n}{n-d_{\max}} \cdot \frac{d_{\min} \cdot (1-c_1)}{e^{d_{\min} \cdot (1-c_1)-1}} \\
& = e^{(d_{\min}+d_{\max}) \cdot c_1} \cdot \frac{n}{n-d_{\max}} \cdot \frac{d_{\min} \cdot (1-c_1)}{e^{d_{\min}-1}} \\
& \le e^{2 \cdot d_{\max} \cdot c_1} \cdot \frac{n}{n-d_{\max}} \cdot \frac{2}{e} \qquad \text{($1-c_1 \le 1$ and $f(d_{\min}) \le f(2) = 2/e$)} \\
& \le 2 \cdot e^{-0.92} \cdot \frac{1}{0.9} \qquad \left(\frac{n}{n-d_{\max}} \le \frac{1}{0.9} \text{ for sufficiently large n and } c_1 \le \frac{0.04}{d_{\max}}\right) \\
& \le 0.9
\end{align*}
% Now since $c_1 \le \frac{\frac{1}{e}-2 \log\left(\frac{10}{9}\right)}{d_{\max}}$ (note that we can do so because this term is always positive), we have 
% \[
% e^{d_{\max} \cdot c_1 - \frac{1}{e}} \cdot \frac{1}{0.9} \le \frac{9}{10}
% \] 
The desired result follows.
\hfill\qed
\end{proof}

%tcs
Finally, we provide an upper bound on  the remaining term $\frac{L}{n-d_{\max}}$ in the bound of Lemma~\ref{lemm_small_2}.
\begin{lemma}\label{lemm_small_4}
For sufficiently large $n$ and $c_1 \leq 0.2$ we have 
$\frac{L}{n-d_{\max}}\geq 1$.
\end{lemma}
\begin{proof}
We have the following inequality:
\[
\begin{array}{rcl}
L & = & \sum_{i=1}^{x} (a_i-t_i)\cdot d_i \cdot e^{d_i\cdot k/n} \\[1.5ex]
& \geq & 2 \cdot \sum_{i=1}^x (a_i-t_i) \\[1.5ex]
& = & 2 \cdot (n-t) \\[1.5ex]
& \ge & 2 \cdot n \cdot (1-c_1) \\[1.5ex]
& \ge & 1.6 \cdot n,
\end{array}
\]
where the second transition holds because $e^{d_i\cdot k/n}\geq 1$ and $d_i \geq d_{\min}\geq 2$, the fourth transition holds because $t \leq c_1 \cdot n$, and the last transition holds because $c_1 \le 0.2$. Finally, $n-d_{\max} < 1.6 \cdot n$ for large $n$.
Hence, the desired result follows. 
\hfill\qed
\end{proof}

Now we prove %tcs - the 
a 
bound on $b_{k_1,k_2,...,k_x}$.
\begin{lemma}[Upper bound on $b_{k_1,k_2,...,k_x}$]\label{upper_bound_b}
There exists a constant $c_1 > 0$ such that for sufficiently large $n$ and $t\leq k \leq c_1 \cdot n$,
we have $b_{k_1,k_2,...,k_x} \le \left(\frac{9}{10}\right)^k$.
\end{lemma}
\begin{proof}
Let $0< c_1 \le \frac{0.04}{d_{\max}} \le 0.2$ as in Lemma~\ref{lemm_small_3}.
By Lemma~\ref{lemm_small_2} we have 
\[
b_{k_1,k_2,...,k_x}  
\le \left(\frac{L}{n-d_{\max}}\right)^{-t} \cdot 
\left( \frac{L}{n-d_{\max}} \cdot e^{1-\frac{\sum_{i=1}^{x} d_i \cdot (a_i-t_i)}{n}} \right)^k
\]
By Lemma~\ref{lemm_small_4} we have $\left(\frac{L}{n-d_{\max}}\right)\geq 1$,
and hence $\left(\frac{L}{n-d_{\max}}\right)^{-t} \leq 1$.
By Lemma~\ref{lemm_small_3} we have 
\[
\frac{L}{n-d_{\max}} \cdot e^{1-\frac{\sum_{i=1}^{x} d_i \cdot (a_i-t_i)}{n}} \leq \frac{9}{10}
\]
The desired result follows trivially.
\hfill\qed
\end{proof}

Taking appropriate bounds on the value of $k$, we get an upper bound on $a_{k_1,k_2,...,k_x}$. Recall that $x$ is the number of distinct degrees and hence $x \le d_{\max}-d_{\min}+1$.
\begin{lemma}[Upper bound on $a_{k_1,k_2,...,k_x}$]\label{lemm_small_5}
There exists a constant $c_1 > 0$ such that for sufficiently large $n$ with $t\leq c_1 \cdot n$ and for all $30 \cdot x \cdot \log(n) \leq k \leq c_1 \cdot n$,
we have $a_{k_1,k_2,...,k_x} < \frac{1}{n^{3\cdot x}}$.
\end{lemma}
\begin{proof}
By Lemma~\ref{lemm_bound_small_1} we have $a_{k_1,k_2,...,k_x} \leq b_{k_1,k_2,...,k_x}$ and 
by Lemma~\ref{upper_bound_b} we have $b_{k_1,k_2,...,k_x} \le \left(\frac{9}{10}\right)^k$. Thus for $k \ge 30 \cdot x \cdot \log(n)$,
\[
a_{k_1,k_2,...,k_x} \le \left(\frac{9}{10}\right)^{30 \cdot x \cdot \log(n)} 
=n^{30\cdot x \cdot \log(9/10)} 
\leq \frac{1}{n^{3\cdot x}}
\]
The desired result follows.
\hfill\qed
\end{proof}

\begin{lemma}[Main lemma for small $k$]\label{lemm_small_main}
There exists a constant $c_1 > 0$ such that for sufficiently large $n$ with $t\leq c_1 \cdot n$ and for all $30 \cdot x \cdot \log(n) \leq k \leq c_1 \cdot n$,
the probability that the size of the reverse reachable set $S$ is $k$ is at most 
$\frac{1}{n^2}$.
\end{lemma}
\begin{proof}
The probability that the reverse reachable set is of size $k$ is given by $\alpha_k$. 
By Lemma~\ref{lemm_bound_3} and Lemma~\ref{lemm_small_5} it follows that the probability is at most 
$n^x \cdot n^{-3\cdot x} =n^{-2\cdot x} \leq \frac{1}{n^2}$.
The desired result follows.
\hfill\qed
\end{proof}

\begin{comment}
Putting $k \ge 7 log(n)$, we get 
\begin{equation*}
 b_{k_1,k_2,...,k_x} < \left(\frac{2}{e}\right)^{7 log(n)} = n^{7 log\left(\frac{2}{e}\right)} < n^{-2}
\end{equation*}

For finding $t$, where for $k \ge t$ we have $b_{k_1,k_2,...,k_x} < n^{-2}$, we put $k=t$ in the maximum value of $b_{k_1,k_2,...,k_x}$ obtained.
\begin{eqnarray*}
 b_{k_1,k_2,...,k_x} && = e^{-t \frac{\sum_{i=1}^{x} d_i (a_i-t_i)}{n}} \\
 && = e^{-t \frac{\sum_{i=1}^{x} d_i (a_i-t_i)}{n-t} (1-\frac{t}{n})} \\
 && \le e^{-2 t (1-\frac{t}{n})}
\end{eqnarray*}
Clearly, for $t > log(n)$ and $t << n$, we have $b_{k_1,k_2,...,k_x} < n^{-2}$.

{\bf $\clubsuit$ 
KRISH: What is the main lemma of this section. we have bounded $b_{...}$. What should be 
the main lemma. How does it follow from bound of $b$. 
Doubt there is the term $\sum_{\sum k_i=k, k_i \geq t_i}$. Why is it enough to show $a_{...} < 1/n^2$. 
$\clubsuit$ 
}

\end{comment}

%
% ---------------------------------------------------------------------------------%
%
\subsection{Large $k$: $c_1 \cdot n \le k \le c_2 \cdot n $}
\label{sec:large_k}
In this section we will show that for all constants $c_1$ and $c_2$, with $0 < c_1\leq c_2$, 
when $t \leq c_2 \cdot n$ the probability $\alpha_k$ is at most $\frac{1}{n^2}$ for all 
$c_1 \cdot n \le k \leq c_2 \cdot n$.
We start with some notation that we will use in the proofs.
Let $a_i = p_i \cdot n, t_i = y_i \cdot n, k_i= s_i\cdot n$ for 
$1 \le i \le x$ and $k = s\cdot  n$ for $c_1 \le s < c_2$.
We first present a bound on $a_{k_1,k_2,\ldots,k_x}$.

\newcommand{\TM}{\mathsf{Term}}

\begin{lemma}\label{lemm_large_1}
For all constants $c_1$ and $c_2$ with $0 < c_1 \le c_2$ and for all $c_1 \cdot n \le k \le c_2 \cdot n$, we have 
\[
a_{k_1,k_2,\ldots,k_x} \leq (n+1)^x \cdot \TM_1 \cdot \TM_2,
\]
where  
\[
\TM_1 =\left(\prod_{i=1}^{x} \left(\frac{p_i-y_i}{s_i-y_i}\right)^{s_i-y_i} \left(\frac{p_i-y_i}{p_i-s_i}\right)^{p_i-s_i} (1-s)^{d_i (p_i-s_i)} (1-(1-s)^{d_i})^{s_i-y_i} \right)^n
\]
and 
\[
\TM_2= \prod_{i=1}^{x} \left(\frac{1-\left(1-\frac{s}{1-d_i/n}\right)^{d_i}}{1-\left(1-s\right)^{d_i}}\right)^{n(s_i-y_i)}.
\]
\end{lemma}
\begin{proof}
We have 
\begin{eqnarray*}
a_{k_1,k_2,\ldots,k_x}= && \left(\prod_{i=1}^{x} \binom{a_i-t_i}{k_i-t_i} \cdot {\left(\frac{\binom{n-k}{d_i}}{\binom{n}{d_i}}\right)}^{a_i-k_i}\right) \cdot R(k_1,k_2,...,k_x) \\
&& \le \left(\prod_{i=1}^{x} (a_i-t_i+1) \cdot \left(\frac{a_i-t_i}{k_i-t_i}\right)^{k_i-t_i} \cdot \left(\frac{a_i-t_i}{a_i-k_i}\right)^{a_i-k_i} {\left(\frac{\binom{n-k}{d_i}}{\binom{n}{d_i}}\right)}^{a_i-k_i}\right) 
\cdot R(k_1,k_2,...,k_x) \\
&& \qquad \text{(Applying second inequality of Proposition~\ref{prop_approx} with $\ell=a_i-t_i$ and $j=k_i-t_i$)} \\
&& \le (n+1)^x \cdot \left(\prod_{i=1}^{x} \left(\frac{a_i-t_i}{k_i-t_i}\right)^{k_i-t_i} \cdot \left(\frac{a_i-t_i}{a_i-k_i}\right)^{a_i-k_i} {\left(\frac{\binom{n-k}{d_i}}{\binom{n}{d_i}}\right)}^{a_i-k_i}\right) 
\cdot R(k_1,k_2,...,k_x). 
\end{eqnarray*}
Proposition~\ref{prop_approx} is presented in the technical appendix.
The last inequality above is obtained as follows: $(a_i-t_i+1) \leq n + 1$ as $a_i \leq n$. %%and since $R(k_1,k_2,\ldots,k_x)$ is a probability value it is at most~1.
Our goal is now to show that 
\[
Y= \left(\prod_{i=1}^{x} \left(\frac{a_i-t_i}{k_i-t_i}\right)^{k_i-t_i} \cdot \left(\frac{a_i-t_i}{a_i-k_i}\right)^{a_i-k_i} {\left(\frac{\binom{n-k}{d_i}}{\binom{n}{d_i}}\right)}^{a_i-k_i}\right) 
\cdot R(k_1,k_2,\ldots,k_x) 
\leq \TM_1 \cdot \TM_2.
\]
We have (i)~$a_i-t_i=n(p_i-y_i)$; (ii)~$k_i-t_i=n(s_i-y_i)$; and 
(iii)~$a_i-k_i=n(p_i-s_i)$.
Hence we have
\[
\prod_{i=1}^{x} \left(\frac{a_i-t_i}{k_i-t_i}\right)^{k_i-t_i} \cdot \left(\frac{a_i-t_i}{a_i-k_i}\right)^{a_i-k_i} 
= \prod_{i=1}^{x} \left(\frac{p_i-y_i}{s_i-y_i}\right)^{n(s_i-y_i)} \left(\frac{p_i-y_i}{p_i-s_i}\right)^{n(p_i-s_i)}. 
\]
By Lemma~\ref{lemm_bound_2} we have
\[
\prod_{i=1}^{x} {\left(\frac{\binom{n-k}{d_i}}{\binom{n}{d_i}}\right)}^{a_i-k_i} 
\leq \prod_{i=1}^x {\left(1-\frac{k}{n}\right)}^{d_i \cdot n\cdot (p_i-s_i)} 
\]
By Lemma~\ref{lemm_upper_bound} we have
\[
R(k_1,k_2,\ldots,k_x) \leq \prod_{i=1}^x
\left(1-\left(1-\frac{k}{n-d_i}\right)^{d_i}\right)^{n(s_i-y_i)}
\]
Hence we have
\begin{eqnarray*}
Y 
&& \le \prod_{i=1}^{x} \left(\frac{p_i-y_i}{s_i-y_i}\right)^{n(s_i-y_i)} \left(\frac{p_i-y_i}{p_i-s_i}\right)^{n(p_i-s_i)} {\left(1-\frac{k}{n}\right)}^{d_i n (p_i-s_i)} \left(1-\left(1-\frac{k}{n-d_i}\right)^{d_i}\right)^{n(s_i-y_i)} \\
&& = \prod_{i=1}^{x} \left(\frac{p_i-y_i}{s_i-y_i}\right)^{n(s_i-y_i)} \left(\frac{p_i-y_i}{p_i-s_i}\right)^{n(p_i-s_i)} {\left(1-s\right)}^{d_i n (p_i-s_i)} \left(1-\left(1-\frac{s}{1-d_i/n}\right)^{d_i}\right)^{n(s_i-y_i)} \\
&& = \underbrace{\left(\prod_{i=1}^{x} \left(\frac{p_i-y_i}{s_i-y_i}\right)^{s_i-y_i} \left(\frac{p_i-y_i}{p_i-s_i}\right)^{p_i-s_i} (1-s)^{d_i (p_i-s_i)}\right)^n}_{X_1} \ \cdot \  %\\&& 
\underbrace{\prod_{i=1}^{x} \left(1-\left(1-\frac{s}{1-d_i/n}\right)^{d_i}\right)^{n(s_i-y_i)}}_{X_2} \\
&& = \left(\prod_{i=1}^{x} \left(\frac{p_i-y_i}{s_i-y_i}\right)^{s_i-y_i} \left(\frac{p_i-y_i}{p_i-s_i}\right)^{p_i-s_i} (1-s)^{d_i (p_i-s_i)} (1-(1-s)^{d_i})^{s_i-y_i} \right)^n \\
&& \quad \cdot \prod_{i=1}^{x} \left(\frac{1-\left(1-\frac{s}{1-d_i/n}\right)^{d_i}}{1-\left(1-s\right)^{d_i}}\right)^{n(s_i-y_i)}
\end{eqnarray*}
The last equality is obtained by multiplying $(1-(1-s)^{d_i})^{n(s_i-y_i)}$ to $X_1$ and dividing it from $X_2$.
Thus we obtain $Y \le \TM_1 \cdot \TM_2$, and the result follows.
\hfill\qed
\end{proof}

%tcs
Given the bound in Lemma~\ref{lemm_large_1}, we now present upper bounds on $\TM_2$ and $\TM_1$. 
\begin{lemma}\label{lemm_large_2}
$\TM_2$ of Lemma~\ref{lemm_large_1}, i.e.,
$\displaystyle\prod_{i=1}^{x} \left(\frac{1-\left(1-\frac{s}{1-d_i/n}\right)^{d_i}}
{1-\left(1-s\right)^{d_i}}\right)^{n(s_i-y_i)}$ 
is bounded from above by a constant.
\end{lemma}
\begin{proof}
We have  
\begin{eqnarray*}
\left(\frac{1-\left(1-\frac{s}{1-d_i/n}\right)^{d_i}}{1-\left(1-s\right)^{d_i}}\right)^{n(s_i-y_i)}  %\\
\le && \left(\frac{1-\left(1-s(1+\frac{2d_i}{n})\right)^{d_i}}{1-\left(1-s\right)^{d_i}}\right)^{n(s_i-y_i)} \qquad \quad \text{(for sufficiently large $n$)}\\
\le && \left(\frac{1-(1-s)^{d_i}+\binom{d_i}{1} \cdot \frac{2sd_i}{n}\cdot (1-s)^{d_i-1}}{1-(1-s)^{d_i}}\right)^{n(s_i-y_i)} \\
&& \qquad \quad \text{(taking first two terms of bionomial expansion)}\\
= && \left(1+\frac{ \frac{(1-s)^{d_i-1}}{1-(1-s)^{d_i}} \cdot 2s {d_i}^2 }{n}\right)^{n(s_i-y_i)}\\
\le && e^{\frac{(1-s)^{d_i-1}}{1-(1-s)^{d_i}} \cdot 2s {d_i}^2 \cdot (s_i-y_i)} \qquad \quad\text{($(1+x) \le e^x$)}.
\end{eqnarray*}
Since $c_1 \le s \le c_2$ we have $s$ is constant, and similarly $d_{\min} \le d_i \le d_{\max}$ and
hence $d_i$ is constant.
Hence it follows that the above expression is constant and hence 
the product of those terms for $1 \le i \le x$ is also bounded by a constant (since $x$ is constant).
The result follows.
\hfill\qed
\end{proof}

\begin{lemma}\label{lemm_large_3}
There exists a constant $0<\eta<1$ such that $\TM_1$ of Lemma~\ref{lemm_large_1} is at most 
$\eta^n$ (exponentially small), i.e., 
\[
\left(\prod_{i=1}^{x} \left(\frac{p_i-y_i}{s_i-y_i}\right)^{s_i-y_i} \left(\frac{p_i-y_i}{p_i-s_i}\right)^{p_i-s_i} (1-s)^{d_i (p_i-s_i)} (1-(1-s)^{d_i})^{s_i-y_i} \right)^n
\leq \eta^n
\]
\end{lemma}
\begin{proof}
Let 
\begin{equation*}
f(d_i) = \left(\frac{p_i-y_i}{s_i-y_i}\right)^{s_i-y_i} \left(\frac{p_i-y_i}{p_i-s_i}\right)^{p_i-s_i} (1-s)^{d_i (p_i-s_i)} (1-(1-s)^{d_i})^{s_i-y_i}
\end{equation*}
Note that $f(d_i)$ is maximum when $$\partial_{d_i} f(d_i) = 0 \Leftrightarrow d_i^{*} = \frac{\log\left(\frac{p_i-s_i}{p_i-y_i}\right)}{\log(1-s)}$$
Moreover, it can easily be checked that this maximum value is $f(d_i^{*}) = 1$. Hence, in general we have $f(d_i) \le 1$. We wish to prove that there exists some $i$ such that $d_i \neq d_i^*$. Suppose for contradicton that $d_i = d_i^*$ for all $i$. Then, we have 
$$
d_i^* \ge 2 \Rightarrow (1-s)^2 \ge \frac{p_i-s_i}{p_i-y_i}
$$
for all $i$. For fractions $\alpha_i/\beta_i$, we have $(\sum_i \alpha_i)/(\sum_i \beta_i) \le \max_i \alpha_i/\beta_i$. Hence, we have 
$$
(1-s)^2 \ge \frac{\sum_i (p_i-s_i)}{\sum_i (p_i-y_i)} = \frac{1-s}{1-y} \Rightarrow (1-s)(1-y) \ge 1
$$
The last inequality is a contradiction, because $0 < s < 1$. Hence, not all $d_i$ can be equal to $d_i^*$. Hence, $\prod_i f(d_i)$ cannot achieve its maximum value $1$. Since each $d_{i^*} \in [d_{\min},d_{\max}]$ has a compact domain and $f$ is a continuous function, there exists a constant $\eta < 1$ such that $\prod_{i=1}^x f(d_i) \leq \eta$. The result thus follows.
\hfill\qed
\end{proof}

\begin{lemma}[Main lemma for large $k$]\label{lemm_main_large}
For all constants $c_1$ and $c_2$ with $0 < c_1 \leq c_2$, when $n$ is sufficiently large and $t \le c_2 \cdot n$, 
for all $c_1 \cdot n \leq k \leq c_2 \cdot n$, the probability that the size of the reverse reachable 
set $S$ is $k$ is at most $\frac{1}{n^2}$.
\end{lemma}
\begin{proof}
By Lemma~\ref{lemm_large_1}, we have $a_{k_1,k_2,\ldots,k_x}\leq (n+1)^x \cdot \TM_1\cdot \TM_2$, and 
by Lemma~\ref{lemm_large_2} and Lemma~\ref{lemm_large_3}, 
$\TM_2$ is a constant and $\TM_1$ is exponentially small in $n$,
where $x \le (d_{\max}-d_{\min}+1)$.
The exponentially small $\TM_1$ overrides the polynomial factor $(n+1)^x$ and the constant $\TM_2$,
and ensures that  $a_{k_1,k_2,\ldots,k_x} \leq n^{-3x}$.
By Lemma~\ref{lemm_bound_3} it follows that $\alpha_k \leq n^{-2x} \leq \frac{1}{n^2}$.
\hfill\qed
\end{proof}

%This proves that the upper bound decreases exponentially in $n$ (overriding the polynomial factor and constant in front and end of the equation respectively) and hence can be made less than $1/n^2$ for sufficiently large $n$.
%
%----------------------------------------------------------------------------------------------------------------------------------------
%
\subsection{Very large $k$: $(1-1/e^2)n$ to $n-d_{\min}-1$}
%%{\bf $\clubsuit$ Nisarg: This section is completely re-written.}
In this subsection we consider the case when the size $k$ of the 
reverse reachable set is between 
$(1-\frac{1}{e^2})\cdot n$ and $n-d_{\min}-1$.
Note that if the reverse reachable set has size at least $n-d_{\min}$,
then the reverse reachable set must be the set of all vertices, 
as otherwise the remaining vertices cannot have enough edges among themselves.
Take $\ell=n-k$. Hence $d_{\min}+1 \le \ell \le n/e^2$.
As stated earlier, in this case $a_{k_1,k_2,...,k_x}$ becomes small since we require that the $\ell$ vertices
outside the reverse reachable set must have all their edges within themselves; this corresponds to the factor of 
${\left(\binom{n-k}{d_i}/\binom{n}{d_i}\right)}^{a_i-k_i}$. Since $\ell$ is very small, this has a very low probability.
With this intuition, we proceed to show the following bound on $a_{k_1,k_2,...,k_x}$.

\begin{lemma}\label{lemm_very_large_1}
We have $a_{k_1,k_2,...,k_x} \leq \left(x \cdot e \cdot \frac{\ell}{n}\right)^\ell$.
\end{lemma}
\begin{proof}
We have 
\begin{eqnarray*}
a_{k_1,k_2,...,k_x} = && \left(\prod_{i=1}^{x} \binom{a_i-t_i}{k_i-t_i} {\left(\frac{\binom{n-k}{d_i}}{\binom{n}{d_i}}\right)}^{a_i-k_i}\right) \cdot R(k_1,k_2,...,k_x) \\
&& \le \prod_{i=1}^{x} \binom{a_i-t_i}{k_i-t_i} {\left(\frac{\binom{n-k}{d_i}}{\binom{n}{d_i}}\right)}^{a_i-k_i} \quad \text{(Ignoring probability value $R(k_1,k_2,\ldots,k_x) \leq 1$)}\\
&& = \prod_{i=1}^{x} \binom{a_i-t_i}{a_i-k_i} {\left(\frac{\binom{n-k}{d_i}}{\binom{n}{d_i}}\right)}^{a_i-k_i} \quad \text{(Since } \binom{x}{y} = \binom{x}{x-y} \text{)}\\
&& \le \prod_{i=1}^{x} \binom{a_i-t_i}{a_i-k_i} {\left(1-\frac{k}{n}\right)}^{d_i(a_i-k_i)} \quad \text{(By Lemma~\ref{lemm_bound_2})} \\
&& \le \prod_{i=1}^{x} \left(\frac{e\cdot (a_i-t_i)}{a_i-k_i}\right)^{a_i-k_i} {\left(\frac{n-k}{n}\right)}^{d_i(a_i-k_i)} \quad \text{(Inequality~1 of Proposition~\ref{prop_approx})}\\
&& \le e^{\ell}\cdot \left(\frac{\ell}{n}\right)^{2\ell} \cdot 
\prod_{i=1}^{x} \left(\frac{a_i-t_i}{a_i-k_i}\right)^{a_i-k_i}
\qquad \text{(Since } d_i \ge 2 \text{ and } \sum_{i=1}^x (a_i-k_i) = \ell \text{)}\\
\end{eqnarray*}
Recall that in the product appearing in the last expression, we take the value of the term to be $1$ where $a_i = k_i$. 
Proposition~\ref{prop_approx} is presented in the technical appendix.
Since for all $i$ we have $(a_i-t_i)\leq n-t$, it follows that  
$\prod_{i=1}^{x}(a_i-t_i)^{a_i-k_i}\le\prod_{i=1}^{x} (n-t)^{a_i-k_i}=(n-t)^\ell$. 

We also want a lower bound for $\prod_{i=1}^{x} (a_i-k_i)^{a_i-k_i}$. 
Note that $\sum_{i=1}^x (a_i-k_i) = \ell$ is fixed. Hence, this is a problem of minimizing $\prod_{i=1}^x {y_i}^{y_i}$ given that $\sum_{i=1}^x y_i = \ell$ is fixed. 
As before, this reduces to $\partial_{y_a} \prod_{i=1}^x {y_i}^{y_i} = \partial_{y_b} \prod_{i=1}^x {y_i}^{y_i}$, for all $a,b$. 
Hence, the minimum is attained at $y_i = \ell/x$,  for all $i$. 
Hence, $\prod_{i=1}^{x} (a_i-k_i)^{a_i-k_i} \ge \left(\frac{\ell}{x}\right)^\ell$. 
Combining these,
\begin{eqnarray*}
a_{k_1,k_2,...,k_x} && \le  e^{\ell}\cdot \left(\frac{\ell}{n}\right)^{2\ell} \cdot 
\prod_{i=1}^{x} \left(\frac{a_i-t_i}{a_i-k_i}\right)^{a_i-k_i} \\
&& \le e^{\ell} \cdot \left(\frac{\ell}{n}\right)^{2\ell} \cdot \left(\frac{n-t}{\left(\frac{\ell}{x}\right)}\right)^\ell\\
&& \le \left(x\cdot e \cdot \frac{\ell}{n}\right)^\ell
\end{eqnarray*}
Hence we have the desired inequality.
\hfill\qed
\end{proof}

We see that $\left(x \cdot e\cdot \frac{\ell}{n}\right)^\ell$ is a convex function in $\ell$ and its maximum is attained at one of the endpoints. 
For $\ell=n/e^2$, the bound is exponentially decreasing with $n$ whereas for constant $\ell$, the bound is polynomially decreasing in $n$. 
Hence, the maximum is attained at left endpoint of the interval (constant value of $\ell$).
However, the bound we get is not sufficient to apply Lemma~\ref{lemm_bound_3} directly. 
%An important observation is that as $\ell$ becomes smaller and smaller, the number of combinations $\sum k_i = k$, where $t_i \le k_i \le a_i$ in the expression of $\alpha_k$ also decrease. Thus, 
We break this case into two sub-cases; $d_{\max}+1 < \ell \le n/e^2$ and $d_{\min}+1 \le \ell \le d_{\max}+1$.

\begin{lemma}\label{lemm_very_large_2}
 For $d_{\max}+1 < \ell \le n/e^2$, we have $a_{k_1,k_2,...,k_x} < n^{-(2+x)}$ and $\alpha_k \le 1/n^2$.
\end{lemma}
\begin{proof}
 As we have seen, we only need to prove this for the value of $\ell$ where $a_{k_1,k_2,...,k_x}$ attains its maximum i.e. $\ell = d_{\max}+2$. 
Note that $d_{\max}+1 = x+d_{\min} \ge x+2$. Hence, 
\begin{eqnarray*}
 a_{k_1,k_2,...,k_x} && \le \left(x \cdot e \cdot \frac{\ell}{n}\right)^\ell \qquad\text{(By Lemma~\ref{lemm_very_large_1})}\\
 && \le \left(x \cdot e \cdot \frac{d_{\max}+2}{n}\right)^{d_{\max}+2} \\
 && =(x \cdot e \cdot (d_{\max}+2))^{d_{\max}+2} \cdot  n^{-(d_{\max}+2)} \\
 && < n^{-(d_{\max}+1)} \qquad \mbox{ (Since first term is a constant) }\\
 && \le n^{-(2+x)}
\end{eqnarray*}
Hence we obtain the first inequality of the lemma.
By Lemma~\ref{lemm_bound_3} and the first inequality of the lemma we have 
$\alpha_k \leq \frac{1}{n^2}$.
\hfill\qed
\end{proof}

\begin{lemma}\label{lemm_very_large_3}
 There exists a constant $h > 0$ such that for $d_{\min}+1 \le \ell \le d_{\max}+1$, we have $a_{k_1,k_2,...,k_x} < h \cdot n^{-\ell}$ and  $\alpha_k \le \frac{h}{n^2}$.
\end{lemma}
\begin{proof}
By Lemma~\ref{lemm_very_large_1} we have 
\begin{eqnarray*}
 a_{k_1,k_2,...,k_x} && \le \left(x\cdot e \cdot \frac{\ell}{n}\right)^\ell \\
 && \le (x \cdot e \cdot (d_{\max}+1))^{d_{\max}+1} \cdot n^{-\ell}
\end{eqnarray*}
 Let $h = (x \cdot e \cdot (d_{\max}+1))^{d_{\max}+1}$. Hence, first part is proved.

Now, for the second part, we note that since there are $\ell$ vertices outside the reverse 
reachable set, and all their edges must be within these $\ell$ vertices, they must have degree at most $\ell-1$. 
Hence, there are now $n$ vertices with at most $\ell-d_{\min}$ distinct degrees. 
Hence, in the summation 
$$
\alpha_k=\sum_{\substack{k_1,\ldots,k_x \text{ s.t.}\\ \sum k_i =k, t_i \le k_i \le a_i}} a_{k_1,k_2,\ldots,k_x},
$$ 
there are at most $n^{\ell-d_{\min}}$ terms. 
Thus we have
$$\alpha_k \le n^{\ell-d_{\min}} \cdot h \cdot n^{-\ell} = h \cdot n^{-d_{\min}} \le \frac{h}{n^2}.$$
The desired result follows.
\hfill\qed
\end{proof}

\begin{lemma}[Main lemma for very large $k$]\label{lemm_main_very_large}
For all $t$, for all $(1-\frac{1}{e^2})\cdot n \leq k \leq  n-1$,
the probability that the size of the reverse reachable set $S$ is $k$ is at most 
$O(\frac{1}{n^2})$.
\end{lemma}
\begin{proof}
By Lemma~\ref{lemm_very_large_2} and Lemma~\ref{lemm_very_large_3} we obtain 
the result for all $(1-\frac{1}{e^2})\cdot n \leq k \leq  n-d_{\min}-1$.
Since the reverse reachable set must contain all vertices if it has size at least $n-d_{\min}$, the result follows.
\hfill\qed
\end{proof}

\begin{comment}
{\bf $\clubsuit$ Nisarg: Either this $h$ can be made up for (perhaps) by the difference of $c_1$ and $c_2$ where we have exponentially decreasing upper bound. 
Or it can be left as it is. The running time will still be linear.}

{\bf KRISH QUESTION: WHERE DO WE SHOW THE CASE FOR $t \geq 30\cdot x\cdot \log(n)$ that the probability of all vertex
reaching target is at least $1-\frac{1}{n}$.}

{\bf KRISH SIMPLIFIED QUESTION: TO GET INTUITITION ABT THE PREV QUESTION. What happens if we fix $d_{\min}=d_{\max}=2$,
and $t=\sqrt{n}$ and want to bound probability for $k$ between $\sqrt{n}$ and $c_1 \cdot n$ for some constant $c_1$.
Is it easy to show $\frac{1}{n^2}$ bound then.}
\end{comment}
%
%----------------------------------------------------------------------------------------------------------------------------------------
%
%
%----------------------------------------------------------------------------------------------------------------------------------------
%
\subsection{Expected Number of Iterations and Running Time}

From Lemma~\ref{lemm_small_main}, Lemma~\ref{lemm_main_large},
and Lemma~\ref{lemm_main_very_large}, we obtain that there 
exists a constant $h$ such that 
\[
\begin{array}{rcl}
\alpha_k &\le & 
\displaystyle 
\frac{1}{n^2},  \qquad 30\cdot x\cdot \log(n) \le k  < n-d_{\max}-1 \\[1ex]
\alpha_k & \le & 
\displaystyle 
\frac{h}{n^2}, \qquad n-d_{\max}-1  \le k  \le n-d_{\min}-1 \\[1ex]
\alpha_k & = & 0                \qquad \qquad n-d_{\min}    \le k  \le n-1
\end{array}
\]
Hence using the union bound we get the following result
\begin{lemma}[Lemma for size of the reverse reachable set]\label{lemm_reverse_reachable_size}
$\Prb(|S|< 30 \cdot x \cdot \log(n) \text{ or } |S|=n) \ge 1-\frac{h}{n}$,
where $S$ is the reverse reachable set of target set
(i.e., with probability at least $1-\frac{h}{n}$ either at most 
$30\cdot x\cdot \log(n)$ vertices reach the target set or all the vertices
reach the target set).
\end{lemma}
\begin{proof}
\begin{align*}
\Prb(|S|< 30 \cdot x \cdot \log(n) \text{ or } |S|=n) &= 1 - \Prb(30 \cdot x \cdot \log(n) \leq |S| \leq n - 1) \\
&\geq 1 - \frac{n-d_{\text{max}}-1}{n^2} - \frac{h(d_{\text{max}} - d_{\text{max}})}{n^2} - 0 \\
&\geq 1 - \frac{h(n-d_{\text{max}}-1)}{n^2} - \frac{h(d_{\text{max}} - d_{\text{max}})}{n^2} \\
&\geq 1 - \frac{hn}{n^2} \\
&= 1 - \frac{h}{n}
\end{align*}
\hfill\qed
\end{proof}

In addition, we note that the number of iterations of the classical algorithm is bounded by the size of the reverse reachable set, because after the first iteration, the graph is reduced to the sub-graph induced by the reverse reachable set. Let $I(n)$ and $T(n)$ denote the expected number of iterations and the 
expected running time of the classical algorithm for MDPs on random graphs with 
$n$ vertices and constant out-degree. 
Then from above we have
\[
I(n) \le \left(1-\frac{h}{n} \right) \cdot 30 \cdot x \cdot \log(n) + \frac{h}{n} \cdot n
\]
It follows that $I(n)=O(\log(n))$. 
For the expected running time we have
\[ 
T(n) \le \left(1-\frac{h}{n} \right) \cdot (30 \cdot x \cdot \log(n))^2 + 
\frac{h}{n} \cdot n^2
\]
It follows that $T(n)=O(n)$. Hence we have the following theorem.

\begin{theorem}
\label{thm_main_const_outdeg}
The expected number of iterations and the expected running time of the 
classical algorithm for MDPs with B\"uchi objectives over graphs  
with constant out-degree are $O(\log(n))$ and $O(n)$, respectively. 
\end{theorem}

\begin{remark}\label{remark_related_model}
For Theorem~\ref{thm_main_const_outdeg}, we considered the model where the out-degree of each vertex $v$ is fixed as $d_v$ and there exist constants $d_{\min}$ and $d_{\max}$ such that $d_{\min} \le d_v \le d_{\max}$ for every vertex $v$. We discuss the implication of Theorem~\ref{thm_main_const_outdeg} for related models. First, when the out-degrees of all vertices are same and constant (say $d^*$), Theorem~\ref{thm_main_const_outdeg} can be applied with the special case of $d_{\min} = d_{\max} = d^*$. A second possible alternative model is when the outdegree of every vertex is a distribution over the range $[d_{\min},d_{\max}]$. Since we proved that the average case is linear for every possible value of the outdegree $d_v$ in $[d_{\min},d_{\max}]$ for every vertex $v$ (i.e., for all possible combinations), it implies that the average case is also linear when the outdegree is a distribution over $[d_{\min},d_{\max}]$.
\end{remark}

% Note that in the above analysis, we considered ``basic'' classes of MDPs where in each class 
% the outdegree of each vertex is fixed. We showed that the average case running time of the 
% classical algorithm is linear for each of these classes. At first, it might seem unreasonable to 
% assume fixed outdegree for each vertex within a class. However, Theorem~\ref{thm_main_const_outdeg} 
% implies that the average case running time of the classical algorithm is 
% also linear for any class that can be obtained as a distribution over these basic classes. 
% For example, the class of MDPs where the outdegree of every vertex is chosen according to 
% some distribution over a range bounded by constants can be obtained as a union of several basic classes 
% and hence the average case running time of the classical algorithm would also be linear for this class. 
% Note that Theorem~\ref{thm_main_const_outdeg} also applies to the special basic class of MDPs where the 
% outdegrees of all the vertices are same and constant.

\begin{comment}
In fact, for the particular case of $t > log(n)$,
\begin{equation}
 Pr(|S|=n-1) \ge 1-\frac{1}{n}
\end{equation}
When $|S|=n-1$, the number of iterations required is $1$. Hence the expected number of iterations $I(n)$ in this case is
\begin{equation}
I(n) \le \left(1-\frac{1}{n} \right) 1 + \frac{1}{n} n = 1
\end{equation}
Thus average case running time $R(n) = O(n)$.
\end{comment}
%
%----------------------------------------------------------------------------------------------------------------------------------------
%

\section{Average Case Analysis in Erd\"os-R\'enyi Model}
In this section we consider the classical Erd\"os-R\'enyi model 
of random graphs $\calg_{n,p}$, with $n$ vertices, where 
each edge is chosen to be in the graph independently with probability 
$p$~\cite{ER} (we consider directed graphs and then $\calg_{n,p}$
is also referred as $\cald_{n,p}$ in the literature).
First, in Section~\ref{subsec_log} we consider the case when 
$p$ is $\Omega\left(\frac{\log(n)}{n}\right)$, and then we consider the case when 
$p=\frac{1}{2}$ (that generates the uniform distribution over all graphs).
We will show two results: 
(1)~if $p \ge \frac{c \cdot \log(n)}{n}$, for some constant $c>2$, then the 
expected number of iterations is constant and the expected running time is 
linear; 
and (2)~if $p=\frac{1}{2}$ (with $p=\frac{1}{2}$ we consider all graphs to be equally likely), 
then the probability that the number of iterations is more than one falls exponentially
in $n$ (in other words, graphs where the running time is more than linear are exponentially
rare).

\subsection{$\calg_{n,p}$ with $p=\Omega\left(\frac{\log(n)}{n}\right)$}\label{subsec_log} 
In this subsection we will show that given $p \ge \frac{c\cdot \log(n)}{n}$, for some constant 
$c>2$, the probability that not all vertices 
can reach the given target set is $O(1/n)$.
Hence the expected number of iterations of the classical algorithm for MDPs 
with B\"uchi objectives is constant and hence the algorithm 
works in average time linear in the size of the graph.
Observe that to show the result the worst possible case is when the size of the target set
is~1, as otherwise the chance that all vertices reach the target set is higher.
Thus from here onwards, we assume that the target set has exactly 1 vertex.
 
%This part develops a bound on the probability of an edge such that not all vertices of a random directed graph can reach a given target vertex. We consider graphs with a given number n of vertices, a single target vertex and with each edge being present independently with a probability $p$, which depends only on $n$. We will show that given $p \ge \frac{8 log(n)}{n}$, the probability that not all vertices can reach the given target vertex is less than $O(1/n$), and hence the classical Algorithm for solving MDPs works on such graphs in average time linear in the number of vertices. We also show that for $p = \frac{1}{2}$. the probability that not all vertices can reach the given target vertex falls exponentially with $n$.

\smallskip\noindent{\em The probability $R(n,p)$.}
For a random graph in $\calg_{n,p}$ and a given target vertex, 
we denote by $R(n,p)$ the probability that each vertex in the graph has a path along 
the directed edges to the target vertex.
Our goal is to obtain a lower bound on $R(n,p)$.

\smallskip\noindent{\em The key recurrence.}
Consider a random graph $G$ with $n$ vertices, with a given target vertex, and 
edge probability $p$.  
For a set $K$ of vertices with size $k$ (i.e., $|K|=k$), which contains the target vertex, 
$R(k,p)$ is the probability that each vertex in the set $K$, has a path to the target vertex, 
that lies within the set $K$ (i.e., the path only visits vertices in $K$). 
The probability $R(k,p)$ depends only on $k$ and $p$, due to the symmetry among vertices.

Consider the subset $S$ of all vertices in $V$, which have a path to the target vertex. 
In that case, for all vertices $v$ in $V\setminus S$, there is no edge going from $v$ to a 
vertex in  $S$ (otherwise there would have been a path from $v$ to the target vertex). 
Thus there are no incoming edges from $V \setminus S$ to $S$. Let $|S| = i$. 
Then the $i\cdot (n-i)$ edges from $V \setminus S$ to $S$ should be absent, and each edge is 
absent with probability $(1-p)$. 
The probability that each vertex in $S$ can reach the target is $R(i,p)$. 
So the probability of $S$ being the reverse reachable set is given by:
\begin{equation}
(1-p)^{i\cdot (n-i)}\cdot R(i,p).
\end{equation}
There are $\binom{n-1}{i-1}$ possible subsets of $i$ vertices that include the given target vertex, and $i$ can range from $1$ to $n$.
Exactly one subset $S$ of $V$ will be the reverse reachable set. 
So the sum of probabilities of the events that $S$ is reverse reachable set is $1$. 
Hence we have:
\begin{equation}
1 = {\sum_{i=1}^{n}}\binom{n-1}{i-1} \cdot (1-p)^{i\cdot(n-i)} \cdot R(i,p)
\end{equation}
Moving all but the last term (with $i=n$) to the other side, we get the following recurrence relation:
\begin{equation}
\label{eqn:recurrence2}
R(n,p) = 1 - {\sum_{i=1}^{n-1}}\binom{n-1}{i-1} \cdot (1-p)^{i\cdot (n-i)} \cdot R(i,p).
\end{equation}

\smallskip\noindent{\em Bound on $p$ for lower bound on $R(n,p)$.}
We will prove a lower bound on $p$ in terms of $n$ such that the probability that not all $n$ vertices 
can reach the target vertex is less than $O(1/n)$. 
In other words,  we require
\begin{equation}
 R(n,p) \ge 1 - O\left(\frac{1}{n}\right)
\end{equation}
Since $R(i,p)$ is a probability value, it is at most~1.
Hence from 
Equation~\ref{eqn:recurrence2} it follows that it suffices to show that 
\begin{equation}
\sum_{i=1}^{n-1} \binom{n-1}{i-1} \cdot (1-p)^{i\cdot (n-i)} \cdot R(i,p) \le
\sum_{i=1}^{n-1} \binom{n-1}{i-1} \cdot (1-p)^{i\cdot (n-i)} 
\le O\left(\frac{1}{n}\right)
\label{eqn:goal}
\end{equation}
to show that $R(n,p) \ge 1- O\left(\frac{1}{n}\right)$.
We will prove a lower bound on $p$ for achieving Equation~\ref{eqn:goal}.
Let us denote by $t_i= \binom{n-1}{i-1}\cdot (1-p)^{i\cdot (n-i)}$, for $1 \le i \le n-1$.
The following lemma establishes a relation of $t_i$ and $t_{n-i}$.

\begin{lemma}
For $1 \le i \le n-1$, we have $t_{n-i} = \frac{n-i}{i} \cdot t_i$. 
\end{lemma}
\begin{proof}
We have
\begin{eqnarray*}
t_{n-i} & = & \binom{n-1}{n-i-1} (1-p)^{i\cdot (n-i)}\\
 & = & \binom{n-1}{i} \cdot (1-p)^{i \cdot (n-i)}\\
 & = & \frac{n-i}{i} \cdot \binom{n-1}{i-1} (1-p)^{i \cdot (n-i)}\\
 & = & \frac{n-i}{i} \cdot t_i
\end{eqnarray*}
The desired result follows.
\hfill\qed
\end{proof}

Define $g_i=t_i + t_{n-i}$, for $1 \le i \le \lfloor n/2 \rfloor$. 
From the previous lemma we have  
\begin{eqnarray*}
g_i=t_{n-i} + t_i & = & \frac{n}{i} \cdot t_i  = \frac{n}{i} \cdot \binom{n-1}{i-1} \cdot (1-p)^{i\cdot (n-i)} 
 =  \binom{n}{i} \cdot (1-p)^{i \cdot (n-i)}.
\end{eqnarray*}
We now establish a bound on $g_i$ in terms of $t_1$.
In the subsequent lemma we establish a bound on $t_1$.

%%{\bf $\clubsuit$ Nisarg: I've changed from here onwards. Next 2 lemmas.}
\begin{lemma}\label{lemm_bound_on_g}
For sufficiently large $n$, if $p \geq \frac{c \cdot \log(n)}{n}$ with $c>2$, then $g_i \le t_1$ for all $2 \le i \le \lfloor\frac{n}{2}\rfloor$. 
\end{lemma} 
\begin{proof} 
Let $p \ge \frac{c \cdot \log(n)}{n}$ with $c>2$. Now
\begin{align*}
\frac{t_1}{g_i} = \frac{(1-p)^{n-1}}{\binom{n}{i} \cdot (1-p)^{i\cdot (n-i)}} & \ge \frac{1}{n^i \cdot (1-p)^{(i-1)\cdot(n-i-1)}} \qquad \text{(Rearranging powers of $(1-p)$ and $\binom{n}{i}\le n^i$)} \\
& \ge \frac{1}{n^i \cdot e^{\frac{-c \cdot \log(n)}{n} \cdot (i-1) \cdot (n-i-1)}} \qquad \text{($1-x \le e^{-x}$)} \\
& = n^{\frac{c}{n} \cdot (i-1) \cdot (n-i-1) - i}
\end{align*}
To show that $t_1 \ge g_i$, it is sufficient to show that for $2 \le i \le \lfloor n/2 \rfloor$,
$$
\frac{c}{n} \cdot (i-1) \cdot (n-i-1) - i \ge 0 \Leftrightarrow \frac{i \cdot n}{(i-1) \cdot (n-i-1)} \le c
$$
Note that $f(i) = \frac{i \cdot n}{(i-1)\cdot(n-i-1)}$ is convex for $2 \le i \le \lfloor n/2 \rfloor$. Hence, its maximum value is attained at either of the endpoints. 
We can see that 
$$
f(2) = \frac{2 \cdot n}{n-3} \le c \qquad \text{(for sufficiently large $n$ and $c>2$)}
$$
and
$$
f(\lfloor n/2 \rfloor) = \frac{\lfloor n/2 \rfloor \cdot n}{(\lfloor n/2 \rfloor-1) \cdot (\lceil n/2 \rceil-1)}
$$
Note that $\lim_{n \rightarrow \infty} f(\lfloor n/2 \rfloor) = 2$, and hence for any constant $c > 2$, $f(\lfloor n/2 \rfloor) \le c$ for sufficiently large $n$.
The result follows.
\hfill\qed
\end{proof}

\begin{lemma}\label{lemm_bound_on_t_1}
For sufficiently large $n$, if $p \geq \frac{c \cdot \log(n)}{n}$ with $c>2$, then $t_1 \leq \frac{1}{n^2}$.
\end{lemma} 
\begin{proof}
We have $t_1 = (1-p)^{n-1}$. For $p \ge \frac{c\cdot \log(n)}{n}$ we have 
\[
\begin{array}{rcl}
t_1 \le \left(1-\frac{c\cdot \log(n)}{n}\right)^{n-1} 
 & \le &
e^{-\frac{c \cdot \log(n)\cdot (n-1)}{n}} \qquad (\text{ Since }1-x \le e^{-x}) \\[2ex]
&\leq  & e^{-2\cdot \log(n)} =\frac{1}{n^2} \text{ (for sufficiently large n, } c>2 \text{)}
\end{array}
\]
Hence, the desired result follows.
\hfill\qed
\end{proof}

We are now ready to establish the main lemma that proves the upper bound 
on $R(n,p)$ and then the main result of the section.

\begin{lemma}\label{lemm_bound_rnp}
For sufficiently large $n$, for all $p \ge \frac{c\cdot \log(n)}{n}$ with $c>2$, we have 
$R(n,p) \geq 1-\frac{1.5}{n}$.
\end{lemma}
\begin{proof}
We first show that $\sum_{i=1}^{n-1} t_i \leq \frac{1.5}{n}$.
We have 
\[
\begin{array}{rcl}
\displaystyle\sum_{i=1}^{n-1} t_i & = &
\displaystyle 
t_1 + t_{n-1} + \sum_{i=2}^{n-2} t_i \\[2ex] 
& \le & 
\displaystyle 
t_1 + t_{n-1} + \sum_{i=2}^{\lfloor n/2 \rfloor} g_i \quad \mbox{($t_{\lfloor n/2 \rfloor}$ is repeated if $n$ is even)}\\[2ex]
& \leq & 
\displaystyle 
n \cdot t_1 + \sum_{i=2}^{\lfloor n/2 \rfloor} g_i \qquad \text{(We apply $t_i + t_{n-i}= \frac{n}{i} \cdot t_i$ with $i=1$)} \\[2ex]
& \leq & 
\displaystyle 
n \cdot t_1 + \sum_{i=2}^{\lfloor n/2 \rfloor} t_1 \qquad \text{(By Lemma~\ref{lemm_bound_on_g} we have $g_i \le t_1$ for $2 \le i \le \lfloor n/2 \rfloor$)} \\[2ex]
& \le & 
\displaystyle 
\frac{3\cdot n}{2}\cdot t_1 \\[2ex]
& \leq & 
\displaystyle 
\frac{3 \cdot n}{2 \cdot n^2} \qquad \qquad \qquad \text{(By Lemma~\ref{lemm_bound_on_t_1} we have $t_1 \leq \frac{1}{n^2}$)} 
\end{array}
\]

By Equation~\ref{eqn:goal} we have that $R(n,p) \geq 1- \sum_{i=1}^{n-1} t_i$. 
It follows that $R(n,p) \geq 1-\frac{1.5}{n}$. 
\hfill\qed
\end{proof}

\begin{theorem}
The expected number of iterations of the classical algorithm for MDPs with
B\"uchi objectives for random graphs $\calg_{n,p}$, with 
$p\geq \frac{c\cdot \log(n)}{n}$, where $c>2$, is $O(1)$, and the average case running 
time is linear.
\end{theorem}
\begin{proof}
By Lemma~\ref{lemm_bound_rnp} it follows that $R(n,p)\geq 1-\frac{1.5}{n}$,
and if all vertices reach the target set, then the classical algorithm
ends in one iteration.
In the worst case the number of iterations of the classical algorithm 
is $n$.
Hence the expected number of iterations is bounded by 
\[
1\cdot\left(1-\frac{1.5}{n}\right) + n \cdot \frac{1.5}{n} = O(1).
\]
Since the expected number of iterations is $O(1)$ and every iteration 
takes linear time, it follows that the average case running time 
is linear.
\hfill\qed
\end{proof}

\subsection{Average-case analysis over all graphs}
In this section, we consider uniform distribution over all graphs, i.e., 
all possible different graphs are equally likely. 
This is equivalent to considering the Erd\"os-R\'enyi model such that each 
edge has probability~$\frac{1}{2}$.
Using $\frac{1}{2} \ge 3\cdot\log(n)/n$ (for $n \ge 17$) and the results from Section~\ref{subsec_log}, 
we already know that the average case running time for $\calg_{n,1/2}$ is linear. 
In this section we show that in $\calg_{n,\frac{1}{2}}$, the probability 
that not all vertices reach the target is in fact exponentially small in $n$.
It will follow that MDPs where the classical algorithm takes more than
constant iterations are exponentially rare.
We consider the same recurrence $R(n,p)$ as in the 
previous subsection and consider $t_k$ and $g_k$ as defined
before.
The following theorem shows the desired result.

\begin{theorem}
In $\calg_{n,\frac{1}{2}}$ with sufficiently large $n$ the probability that the classical 
algorithm takes more than one iteration is less than  
$\left(\frac{3}{4}\right)^n$.
\end{theorem}
\begin{proof}
We first observe that Equation~\ref{eqn:recurrence2} and Equation~\ref{eqn:goal}
holds for all probabilities.
Next we observe that Lemma~\ref{lemm_bound_on_g} holds for $p\ge \frac{c\cdot \log(n)}{n}$ with any constant $c > 2$,
and hence also for $p=\frac{1}{2}$ for sufficiently large $n$.
Hence by applying the inequalities of the proof of Lemma~\ref{lemm_bound_rnp} we
obtain that 
\[
\sum_{i=1}^{n-1} t_i \le \frac{3\cdot n}{2}\cdot t_1.
\]
For $p=\frac{1}{2}$ we have $t_1= \binom{n-1}{0} \cdot \left(1- \frac{1}{2}\right)^{n-1} =\frac{1}{2^{n-1}}$.
Hence we have 
\[
R(n,p) \ge 1- \frac{3\cdot n}{2 \cdot 2^{n-1}} > 1- \frac{1.5^n}{2^n}= 
1-\left(\frac{3}{4}\right)^n.
\]
The second inequality holds for sufficiently large $n$. It follows that the probability that the classical algorithm takes more 
than one iteration is less than $(\frac{3}{4})^n$.
The desired result follows.
\hfill\qed
\end{proof}

\clearpage
\appendix
\section{Technical Appendix}

\begin{proposition}[Useful inequalities from Stirling inequalities]\label{prop_approx}
For natural numbers $\ell$ and $j$ with $j \leq \ell$ we have the following 
inequalities:
\begin{enumerate}

\item $\binom{\ell}{j} \leq \bigg(\frac{e \cdot \ell}{j}\bigg)^j$.

\item $\binom{\ell}{j} \leq (\ell+1)\cdot \left(\frac{\ell}{j}\right)^j 
\cdot \left(\frac{\ell}{\ell-j}\right)^{\ell-j}$.

\end{enumerate}

\end{proposition}
\begin{proof}
The proof of the results is based on the following Stirling inequality for factorial: 
\[
e \cdot \bigg(\frac{j}{e}\bigg)^{j} \le j! \le e \cdot \bigg(\frac{j+1}{e}\bigg)^{j+1}.
\]
We now use the inequality to show the desired inequalities:
\begin{enumerate}

\item We have 
\[
\begin{array}{rcl}
\binom{\ell}{j} & \leq & 
\displaystyle 
\frac{\ell^j}{j!} \; \leq \; \displaystyle \frac{\ell^j \cdot e^j}{e \cdot j^j} \qquad \text{(using Stirling 
inequality)} \\[2ex] 
& \leq & 
\displaystyle 
\frac{1}{e} \cdot \bigg(\frac{e \cdot \ell}{j} \bigg)^j \\
& \leq & 
\displaystyle 
\bigg(\frac{e \cdot \ell}{j} \bigg)^j
\end{array}
\]

\item We have 
\[
\begin{array}{rcl}
\binom{\ell}{j} & = & 
\displaystyle
\frac{\ell !}{j! \cdot (n-j)!} \\[2ex]
& \le & 
\displaystyle
\left(e \cdot \left(\frac{\ell+1}{e}\right)^{\ell+1}\right) \cdot 
\left(\frac{1}{e \cdot \left(\frac{j}{e}\right)^j \cdot e \cdot 
\left(\frac{\ell-j}{e}\right)^{\ell-j} 
}\right)\\[3ex]
& = & 
\displaystyle
\frac{1}{e^2} \cdot 
\left(\ell+1\right)\cdot 
\left(\frac{\ell+1}{j}\right)^j \cdot  
\left(\frac{\ell+1}{\ell-j}\right)^{\ell-j}  \\[2ex]
& \le & 
\displaystyle
\frac{1}{e^2} \cdot 
\left(\ell+1\right)\cdot 
\left(\frac{\ell+1}{\ell}\right)^{\ell}
\left(\frac{\ell}{j}\right)^j \cdot  
\left(\frac{\ell}{\ell-j}\right)^{\ell-j}  \\[2ex]
& \le & 
\displaystyle
\frac{1}{e^2} \cdot 
\left(\ell+1\right)\cdot 
e \cdot 
\left(\frac{\ell}{j}\right)^j \cdot  
\left(\frac{\ell}{\ell-j}\right)^{\ell-j} \quad \left(Since \ \left(1+\frac{1}{\ell}\right)^{\ell} \le e\right)\\[2ex]
& \le & 
\displaystyle
\left(\ell+1\right)\cdot 
\left(\frac{\ell}{j}\right)^j \cdot  
\left(\frac{\ell}{\ell-j}\right)^{\ell-j}  
\end{array}
\]
The first inequality is obtained by applying the Stirling inequality to the 
numerator (in the first term), and applying the Stirling inequality twice 
to the denominator (in the second term).

\end{enumerate}

\hfill\qed
\end{proof}

\end{document}